\newif\iflong
\longtrue

\documentclass[acmsmall,screen,10pt]{acmart} 
\iflong
\setcopyright{none}
\else
\setcopyright{rightsretained}
\fi
\acmJournal{PACMPL}
\acmYear{2020} \acmVolume{4} \acmNumber{OOPSLA} \acmArticle{128} \acmMonth{11} \acmPrice{}\acmDOI{10.1145/3428196}

\iflong
\renewcommand\footnotetextcopyrightpermission[1]{}
\fi

\usepackage{subcaption} %

\usepackage{listings}
\lstdefinelanguage{program}{%
  keywords={%
    let,pass,function,%
    var,const,bool,int,void,atomic,%
    while,do,if,then,else,assume,assert,call,return,rule,forall,with,new,choose,skip,%
    task,async,yield,for,wait,%
    type,relation,init, action, safety, invariant, axiom, input
  },
  morecomment=[l]{//},
  morecomment=[s]{/*}{*/},
  morecomment=[n]{(*}{*)},
  mathescape=true,
  escapeinside=`',
}
\usepackage{booktabs}
\usepackage{mathpartir}
\usepackage{multirow}
\usepackage{amsmath}
\usepackage{amsthm}
\usepackage{latexsym}
\usepackage{tabularx}
\usepackage{wrapfig}
\usepackage{array}
\usepackage{subcaption}
\usepackage[T1]{fontenc}
\usepackage[nodayofweek]{datetime}
\usepackage{graphicx}
\usepackage{paralist}
\usepackage{algorithm}
\usepackage[noend]{algpseudocode}
\usepackage[flushleft]{threeparttable}
\usepackage{framed}
\usepackage[export]{adjustbox}
\usepackage{multicol}
\usepackage{pifont}
\usepackage[xcolor]{changebar}
\cbcolor{magenta}
\usepackage{ragged2e}

\usepackage{xcolor}
\usepackage{extarrows}
\usepackage{mathtools}
\usepackage{paralist}

\usepackage{tikz}
\usetikzlibrary{positioning,fit}

\usepackage[noabbrev,capitalise]{cleveref} %

\newif\ifnitpick
\nitpickfalse

\newif\ifproofs
\proofstrue

\iflong
\newcommand{\refappendix}[1]{\Cref{#1}}
\else
\newcommand{\refappendix}[1]{the extended version~\cite{extendedVersion}}
\fi

\crefformat{section}{\S#2#1#3} %
\crefformat{subsection}{\S#2#1#3}
\crefformat{subsubsection}{\S#2#1#3}
\crefmultiformat{section}{\S\S#2#1#3}%
{ and~#2#1#3}{, #2#1#3}{ and~#2#1#3}

\Crefname{conjecture}{Conjecture}{Conjectures}
\Crefname{proposition}{Proposition}{Propositions}
\Crefname{lemma}{Lemma}{Lemmas}
\Crefname{corollary}{Corollary}{Corollaries}
\Crefname{example}{Example}{Examples}
\Crefname{definition}{Definition}{Definitions}
\Crefname{algorithm}{Alg.}{Alg.}
\Crefname{theorem}{Theorem}{Theorems}
\Crefname{figure}{Figure}{Figures}
\crefname{line}{line}{lines}

\ifproofs
\else
\excludecomment{proof}
\fi

\newtheorem{theo}{Theo}[section] %

\newtheorem{remark}[theo]{Remark}

\let\state\relax  %

\newcommand{\ignore}[1]{}

\newcommand{\powerset}[1]{\mathcal{P}(#1)}
\def\ImportFromMnSymbol#1{%
  \DeclareFontFamily{U} {MnSymbol#1}{}
  \DeclareFontShape{U}{MnSymbol#1}{m}{n}{
   <-6> MnSymbol#15
   <6-7> MnSymbol#16
   <7-8> MnSymbol#17
   <8-9> MnSymbol#18
   <9-10> MnSymbol#19
   <10-12> MnSymbol#110
   <12-> MnSymbol#112}{}
  \DeclareFontShape{U}{MnSymbol#1}{b}{n}{
   <-6> MnSymbol#1-Bold5
   <6-7> MnSymbol#1-Bold6
   <7-8> MnSymbol#1-Bold7
   <8-9> MnSymbol#1-Bold8
   <9-10> MnSymbol#1-Bold9
   <10-12> MnSymbol#1-Bold10
   <12-> MnSymbol#1-Bold12}{}
  \DeclareSymbolFont{MnSy#1} {U} {MnSymbol#1}{m}{n}
}
\newcommand\DeclareMnSymbol[4]{\DeclareMathSymbol{#1}{#2}{MnSy#3}{#4}}
\ImportFromMnSymbol{C}
\DeclareMnSymbol{\diamondminus}{\mathrel}{C}{120}

\newcommand{\isNull}{\mathit{isNull}}

\newcommand{\para}[1]{\vspace{0.3cm}\noindent \textbf{#1.}}

\newcommand{\child}{\operatorname{child}}

\newcommand{\ikey}{\mathit{key}}
\newcommand{\ileft}{\mathit{left}}
\newcommand{\iright}{\mathit{right}}
\newcommand{\repfunc}{\mathcal{A}}

\newcommand{\XleftY}[3][=]{#2.\ileft#1#3}
\newcommand{\XrightY}[3][=]{#2.\iright#1#3}
\newcommand{\XparentY}[3][\in]{#3#1\child(#2)}

\newcommand{\QReachXYK}[3]{#1\searchpath{#3}{#2}}

\newcommand{\PReachXK}[2]{\diamondminus(\QReachXK{#1}{#2})}
\newcommand{\QReachXK}[2]{\rootobj\searchpath{#2}{#1}}
\newcommand{\QReachAbs}[2]{\textit{reach}({#1})}

\newcommand{\QXK}[3][=]{#2.\mathit{key}#1#3}
\newcommand{\QXV}[3][=]{#2.\mathit{tag}#1#3}
\newcommand{\QXKR}[3][=]{#3#1#2.\mathit{key}}
\newcommand{\QXD}[2][]{#1#2.\mathit{del}}

\newcommand{\QXR}[2][]{#1#2.\mathit{rem}}

\DeclareMathOperator*{\PastSymb}{\diamondminus}
\newcommand{\Past}[1]{\diamondminus (#1)}
\newcommand{\Pastlimit}[3]{\mathop{\diamondminus}_{#2}^{#3} (#1)}

\newcommand{\fkey}{\mathit{key}}
\newcommand{\fleft}{\mathit{left}}
\newcommand{\fright}{\mathit{right}}
\newcommand{\fparent}{\mathit{parent}}
\newcommand{\fnext}{\mathit{succ}}
\newcommand{\fprev}{\mathit{pred}}
\newcommand{\fdel}{\mathit{del}}
\newcommand{\frem}{\mathit{rem}}
\newcommand{\fversion}{\mathit{tag}}

\newcommand{\fontCONST}[1]{\mathit{#1}}

\newcommand{\NULL}{\fontCONST{null}}

\newcommand{\noamxr}[1]{}

\newcommand{\polC}[1][]{\mathrm{C}}
\newcommand{\polphi}[1][]{\phi}

\newcommand{\polCsyn}[1][]{\mathtt{C}}
\newcommand{\polphisyn}[1][]{\boldsymbol{\phi}}

\newcommand{\code}[1]{\ensuremath{\mathtt{#1}}}

\newcommand{\wlupdate}[1][]{\operatorname{upd}}
\newcommand{\wlupdateset}[1][]{\operatorname{updSets}}

\newcommand{\exec}{\pi}

\newcommand{\state}{\sigma}

\newcommand{\statedom}{\Sigma} %

\newcommand{\false}{\textit{false}}
\newcommand{\true}{\textit{true}}

\newcommand{\set}[1]{\{{#1}\}}

\newcommand{\rootobj}{{\tt root}}

\newcommand{\searchpath}[1]{\overset{#1}{\leadsto}}

\newcommand{\extend}[1]{\text{extend}_{{#1}}}

\newcommand{\succroot}{\{-\infty\}}
\newcommand{\predroot}{\{+\infty\}}
\newcommand{\QSReachXK}[2]{\succroot\searchpath{#2}{#1}}

\newcommand{\PSReachXK}[2]{\Past{\QSReachXK{#1}{#2}}}
\newcommand{\QSReachX}[1]{\QSReachXK{#1}{}}
\newcommand{\PSReachX}[1]{\Past{\QSReachXK{#1}{}}}

\definecolor{assertion}{rgb}{0, 0, 255}

\newcommand{\oloc}[2]{({#1},{#2})}

\newcommand{\totalto}{\to}
\newcommand{\partialto}{\hookrightarrow}
\newcommand{\nin}{\not\in}

\newcommand{\searchpathgap}[1]{\overset{#1}{\dashrightarrow}}
\newcommand{\QReachXYKgap}[3]{#1\searchpathgap{#3}{#2}}
\newcommand{\QReachXKgap}[2]{\QReachXYKgap{\rootobj}{#1}{#2}}
\newcommand{\PReachXKgap}[2]{\Past{\QReachXKgap{#1}{#2}}}

\newcommand{\QReachXKsucc}[2]{\QReachXK{#1}{{#2}+\epsilon}}
\newcommand{\PReachXKsucc}[2]{\Past{\QReachXKsucc{#1}{#2}}}

\newcommand{\ghost}[1]{\underline{#1}}

\newcommand{\unlocked}[1]{\mathit{unlocked}({#1})}

\newcommand{\fieldloc}{f}

\newcommand{\texbegin}{t_{\star}}
\newcommand{\texend}{t^{\star}} 
\begin{document}
\newif\ifcomments
\commentsfalse
\nochangebars
\definecolor{dg}{cmyk}{0.60,0,0.88,0.27}

\newcommand{\sharonnew}[1]{}
\newcommand{\yotamsmallnew}[1]{}
\newcommand{\artemnew}[1]{}
\newcommand{\adamnew}[1]{}

\ifcomments
\newcommand{\artem}[1]{{\footnotesize\color{olive}[{\bf Artem}: #1]}}
\newcommand{\yotamsmall}[1]{{\footnotesize\color{magenta}[{\bf Yotam}: #1]}}

\newcommand{\sharon}[1]{{\textcolor{blue}{SS: {\em #1}}}}
\newcommand{\adam}[1]{{\textcolor{teal}{AM: {\em #1}}}}
\newcommand{\mooly}[1]{{\textcolor{cyan}{MS: {\em #1}}}}
\newcommand{\noam}[1]{{\textcolor{dg}{NR: {\em #1}}}}
\newcommand{\yotam}[1]{{\textcolor{magenta}{{\bf #1}}}}
\newcommand{\TODO}[1]{{\textcolor{red}{TODO: {\em #1}}}}

\else
\newcommand{\sharon}[1]{}
\newcommand{\adam}[1]{}
\newcommand{\mooly}[1]{}
\newcommand{\neil}[1]{}
\newcommand{\yotam}[1]{}
\newcommand{\TODO}[1]{}
\newcommand{\artem}[1]{}
\newcommand{\yotamsmall}[1]{}

\fi

\newcommand{\commentout}[1]{}
\newcommand{\OMIT}[1]{}

\title{Proving Highly-Concurrent Traversals Correct}

\author[YMY. Feldman]{Yotam M. Y. Feldman}
\affiliation{
  \institution{Tel Aviv University}
  \country{Israel}
}
\email{yotam.feldman@gmail.com}

\author[A. Khyzha]{Artem Khyzha}
\affiliation{
  \institution{Tel Aviv University}
  \country{Israel}
}
\email{artkhyzha@mail.tau.ac.il}

\author[C. Enea]{Constantin Enea}
\affiliation{
  \institution{IRIF, Universit\'{e} de Paris}
  \country{France}
}
\email{cenea@irif.fr}

\author[A. Morrison]{Adam Morrison}
\affiliation{
  \institution{Tel Aviv University}
  \country{Israel}
}
\email{mad@cs.tau.ac.il}

\author[A. Nanevski]{Aleksandar Nanevski}
\affiliation{
  \institution{IMDEA Software Institute}
  \country{Spain}
}
\email{aleks.nanevski@imdea.org}

\author[N. Rinetzky]{Noam Rinetzky}
\affiliation{
  \institution{Tel Aviv University}
  \country{Israel}
}
\email{maon@cs.tau.ac.il}

\author[S. Shoham]{Sharon Shoham}
\affiliation{
  \institution{Tel Aviv University}
  \country{Israel}
}
\email{sharon.shoham@gmail.com}

\begin{abstract}

Modern highly-concurrent search data structures, such as search trees, obtain multi-core scalability and
performance by having operations traverse the data structure without any synchronization.  As a result,
however, these algorithms are notoriously difficult to prove linearizable, which requires identifying
a point in time in which the traversal's result is correct.  The problem is that
traversing the data structure as it undergoes modifications leads to complex behaviors, %
necessitating intricate reasoning about all interleavings of reads by traversals and writes mutating the data structure.

In this paper, we present a general proof technique for proving unsynchronized traversals correct in
a significantly simpler manner, compared to typical concurrent reasoning and prior proof techniques.
Our framework relies only on \emph{sequential properties} of traversals and on a conceptually
simple and widely-applicable condition about the ways an algorithm's writes mutate the data structure.
Establishing that a target data structure satisfies our condition requires only simple concurrent
reasoning,
without considering interactions of writes and reads.  This reasoning can be further
simplified by using our framework.

To demonstrate our technique,
we apply it to prove several interesting and challenging concurrent binary
search trees: the logical-ordering AVL tree, the Citrus tree, and the full contention-friendly
tree.  Both the logical-ordering tree and the full contention-friendly tree are beyond the reach
of previous approaches targeted at simplifying linearizability proofs.

\end{abstract}




\maketitle
\section{Introduction}

A \emph{search data structure} provides a mutable, searchable set or dictionary (e.g., a binary search
tree or B+tree).
Many important systems, such as databases~\cite{Silo,MassTree} and operating systems~\cite{Clements:2012},
rely on highly-concurrent search data structures for multi-core scalability and performance.
In a \emph{highly-concurrent} algorithm, operations synchronize only when modifying the same node.
In particular, traversals searching for a key simply navigate the data structure, without performing synchronization, which
enables them to run completely in parallel on multiple cores.
This design principle is key to search data structure performance~\cite{David:2015:ACS,Gramoli:2015} and
underpins modern concurrent search trees~\cite{Arbel:2014,Brown:2014:GTN,Clements:2012,PFL16:CFTree,Drachsler:2014,Ellen:2010,Howley:2012,Natarajan:2014,intlf},
skip lists~\cite{NoHotSpotSkipList,fraser-phd,OptSkipList}, and lists/hash tables~\cite{HarrisList,HellerHLMMS05,MichaelList,Triplett:2011:RSC}.

Highly-concurrent algorithms are notoriously difficult to prove correct~\cite{phd/Vafeiadis08,PODC10:Hindsight,Lev-AriCK15,FeldmanE0RS18}.
The standard desired correctness condition is \emph{linearizability}~\cite{TOPLAS:HW90}, which requires that every
operation appears to take effect atomically at some point during its execution.
It is hard, however, to identify a point in which a highly-concurrent traversal's result holds, because
traversing the data structure while it is undergoing modifications can lead to following a path whose links
did not exist simultaneously in memory, navigating from a node after it becomes unreachable, and similar
complex behaviors.

Accordingly, an emerging research thrust is to design proof techniques that enable using \emph{sequential}
reasoning to simplify proving the correctness of highly-concurrent algorithms~\cite{Lev-AriCK15,FeldmanE0RS18}.
The vision is for correctness proofs to follow from a meta-theorem about properties of the algorithm's
sequential code, i.e., when running without interference.  The user's job then reduces to proving
that these sequential properties hold, which does not involve difficult \emph{concurrent} reasoning
about interleaved steps of concurrent operations.

While existing work does not yet fully remove concurrent reasoning on the user's part, the
amount of concurrent reasoning required is decreasing.
Whereas base points~\cite{Lev-AriCK15} require the user to prove properties of concurrent traversals,
local view arguments~\cite{FeldmanE0RS18} only rely on sequential properties of traversals, applying under certain conditions, which must be proven with concurrent reasoning.
Unfortunately, the local view framework's preconditions are complex and restrictive, and
are not satisfied by several data structures such as those by~\citet{PFL16:CFTree,Drachsler:2014}.

In this paper, we
present a proof technique based on a \emph{conceptually simpler and widely-applicable condition}, which
enables tackling data structures beyond the reach of previous approaches in addition to simplifying proofs of the data structures
supported by them.

Our technique targets proving \emph{traversal correctness}, a proof goal that was implicit in previous works~\cite{FeldmanE0RS18,PODC10:Hindsight}, which is defined to mean that a traversal searching
for key $k$ reaching a node $n$ implies that at some point during the traversal's execution so far, $n$ satisfied
a \emph{reachability predicate} $P_k$ over the memory state---e.g., that $n$ is on the search path for $k$.
Proving traversal correctness is typically the crux of the data structure's linearizability proof,
which then becomes straightforward to complete by the user.

Our key theorem establishes traversal correctness \emph{without having to reason about how a traversal is affected by concurrent writes},
by reasoning solely about how the algorithm's writes modify the memory state.
Specifically, writes should satisfy a \emph{forepassed} condition, which (informally) states that they do not
reduce the reachability of any memory location unless that location is either not modified later,
\begin{changebar}
or modified to point only to locations that have already been reachable.
\end{changebar}
The use of the forepassed condition \emph{alleviates} the need to reason about \emph{how reads in the traversal interleave with interfering writes}; the condition only depends on how different \emph{writes} interleave.
Proving that the forepassed condition holds can be done by relying (inductively) on traversal correctness,
which greatly simplifies the proof. %
This apparent circularity is valid because our theorem is also proven inductively, so %
both inductions
can be combined (\cref{sec:circular}).

The \emph{only} requirement for applying our technique is that the reachability predicates $P_k$ be compatible with
the traversals, which (informally) means that if a traversal searching for $k$ navigates from node $n$ to node
$n'$ and $P_k(n)$ holds, then so does $P_k(n')$.  Compatibility is a property of the traversal's code on a
static memory state, and so can be established with sequential reasoning.  In fact, it typically holds trivially,
when the predicates are defined based on how traversals navigate the data structure.

Overall, our technique facilitates clear and simple linearizability proofs,
\begin{changebar}
applicable to concurrent search data structures with optimistic traversals implementing sets/maps in sequentially-consistent shared-memory.\footnote{Our framework is also applicable to algorithms whose traversals use stronger synchronization (e.g. hand-over-hand-locking), but this stronger synchronization satisfies stronger properties allowing simpler proofs with alternative methods; see e.g.~\cite{DBLP:conf/popl/AttiyaRR10}.}
\end{changebar}
We demonstrate our technique on
several sophisticated binary search tress (BSTs): the logical ordering AVL tree~\cite{Drachsler:2014} and the
full contention-friendly tree~\cite{PFL16:CFTree}, which cannot be reasoned about with prior proof techniques,
and the Citrus tree~\cite{Arbel:2014}, which is a complex algorithm that has not been proven using
the prior techniques.

\begin{changebar}
A framework~\cite{DBLP:journals/tods/ShashaG88} simplifying proofs of concurrent data structures has recently been pivotal in several works on proof simplification using concurrent separation logic~\cite{DBLP:conf/pldi/KrishnaPSW,DBLP:journals/pacmpl/KrishnaSW18},
but this framework is inapplicable to optimistic traversals (see~\Cref{sec:related}).
\end{changebar}
Our work provides a new proof framework, suitable for highly-concurrent data structures, in which we use sequential reasoning to tackle one of the Gordian knots of this domain.
We believe that the theory underlying our result sheds a light on the
fundamental reasons for the correctness of these algorithms.  We further hope that the theory can be useful for the
design of new algorithms, which would explicitly target satisfying the forepassed condition.

\para{Contributions}
To summarize, this paper makes the following contributions:
\begin{enumerate}
	\item We formally define traversal correctness and present a new general proof argument for the correctness of highly-concurrent traversals. %
	\item We provide a simple condition on interfering writes and prove that it establishes traversal correctness.
	\item We apply our framework to prove several interesting and challenging concurrent data structures, including the logical ordering AVL tree~\cite{Drachsler:2014}, the Citrus tree~\cite{Arbel:2014}, and the full contention-friendly tree~\cite{PFL16:CFTree}.
\end{enumerate}

\section{Background: Proofs of Linearizability}
\label{sec:linearizability-background}
In this section we provide some background on proofs of
linearizability~\cite{TOPLAS:HW90}, which is the standard correctness
criterion for many highly-concurrent data structures.\footnote{Our technique may also be
applicable to proofs of non-linearizable objects~\cite{DBLP:conf/oopsla/SergeyNBD16}, a direction which we
plan to pursue in future work.}
\begin{changebar}
Throughout this paper we assume a sequentially consistent shared-memory system,
i.e., in which the execution is a sequence of interleaved memory operations performed by the threads.
\end{changebar}

Linearizability requires that every individual method invocation appears to take place instantaneously at some point between its invocation and its return. A classic approach to proving linearizability is to show that the concurrent data structure is \emph{simulated} by a reference implementation where methods execute in a \emph{single} step and according to the expected sequential semantics of the data type. Concretely, this corresponds to defining an \emph{abstraction function} (or simulation relation) that relates states of the concurrent data structure with states of a reference sequential implementation of the data type, such that any step in the concurrent data structure is mapped by the abstraction function to a step of the reference implementation (modulo stuttering).

We focus
the presentation on \emph{set} data structures, as are the examples in our paper.
Set data structures implement
the standard methods $\code{insert}(k)$ and $\code{delete}(k)$
for adding or removing an element from the set, respectively,
and $\code{contains}(k)$ which checks membership of an element.
For set data structures, the abstraction function $\repfunc: \statedom \totalto \powerset{\mathbb{N}}$ maps every concrete
memory state $\state \in \statedom$ of the concurrent data structure (a \emph{state} is a mapping from memory locations to values) %
to an abstract (mathematical) set. For simplicity, we assume that set elements are natural numbers.
Establishing linearizability then amounts to showing that in every execution, and
for every method invocation in the execution (also called an \emph{operation}), there is a point during its
execution interval where the operation ``takes effect'' on the
abstract set.  In our context, this means that
in every execution $\exec$ (an execution is a finite sequence of states $\state_0,\state_1,\ldots$):
\begin{itemize}
	\item For every invocation of $\code{contains}(k)$ that returns true, resp., false, there is a state $\state \in \pi$ during the invocation such that $k \in \repfunc(\state)$, resp., $k \not\in \repfunc(\state)$. (Note that both such states may exist, in which case $\code{contains}(k)$ may return either true or false.)%

	\item For every unsuccessful invocation of $\code{insert}(k)$, resp., $\code{delete}(k)$, (i.e., an invocation returning false), there is a state $\state \in \exec$ during the invocation such that $k \in \repfunc(\state)$, resp., $k \not\in \repfunc(\state)$.

	\item For every successful invocation of $\code{insert}(k)$ (returning true), there is an insert-\emph{decisive transition}: a consecutive pair of states $\state_i,\state_{i+1} \in \pi$ during the invocation such that $\repfunc(\state_{i+1}) = \repfunc(\state_i) \cup \set{k}$ and $\repfunc(\state_i) \neq \repfunc(\state_{i+1})$.

	\item For every successful $\code{delete}(k)$ (returning true), there is a delete-\emph{decisive transition}: a consecutive pair of states $\state_i,\state_{i+1} \in \pi$ during the invocation such that $\repfunc(\state_{i+1}) = \repfunc(\state_i) \setminus \set{k}$ and $\repfunc(\state_i) \neq \repfunc(\state_{i+1})$.
\end{itemize}
Further, there is only one point of modification associated with each successful modification: there is a one-to-one mapping between a pair of states $(\state_i,\state_{i+1}) \in \pi$ where $\repfunc(\state_i) \neq \repfunc(\state_{i+1})$ to a method invocation for which $(\state_i,\state_{i+1})$ is a decisive transition (this corresponds to the fact that the methods of the reference implementation execute in a single step).

The above conditions guarantee that all concurrent executions are linearizable.
In the algorithms we consider, the decisive transition of a modifying invocation (successful $\code{insert}$ or $\code{delete}$) can be identified statically to correspond to one fixed write in the method's code, constituting a \emph{fixed linearization point}.
The case of $\code{contains}$ and unsuccessful $\code{insert},\code{delete}$ is different. For such invocations, it is impossible to identify a fixed linearization point, i.e., \emph{statically} identify the state where the abstract set contains an element or not as the state reached when the method executes some \emph{fixed} instruction.
\section{Overview} \label{sec:motivation}
We use the logical-ordering tree~\cite{Drachsler:2014} as a running example for the challenge of proving
linearizability of a highly-concurrent search data structure \begin{changebar}\yotamsmallnew{shortened:}and how our framework simplifies such proofs. %
\end{changebar}

\subsection{Example: The Logical Ordering Tree}
The Logical Ordering (LO) tree is a self-balancing binary search tree (BST), in which keys are stored in both internal and leaf nodes.
To maintain a small height, a self-balancing tree modifies its structure in response to insertions and deletions.
These modifications are performed by \emph{tree rotations}, which mutate the tree's structure without changing
the order of the keys.  \Cref{LO:Rotation} shows an example, in which node $y$ is rotated right so that the
length of the path to the subtree $A$ decreases.

The main challenge for concurrent self-balancing trees is how to avoid having a rotation throw a concurrent traversal
``off track.'' \Cref{LO:Rotation} shows an example, in which a traversal headed towards $A$ that is located at $y$
before the rotation would instead reach $C$, if the rotation happens before the traversal reads $y$'s left child pointer.
The LO tree solves this problem using a \emph{logical ordering} technique.  In addition to the usual $\fleft$/$\fright$
pointers that induce the binary tree structure, each node also has $\fprev$/$\fnext$ pointers, which link the node
to its predecessor/successor, respectively, according to the logical ordering of the keys.  
The idea is that a traversal can follow the $\fprev$/$\fnext$ pointers to find its target node, should some rotation throw it
off track.

\begin{wrapfigure}[13]{r}{.45\textwidth}
\vspace{-5pt}
\includegraphics[scale=0.635]{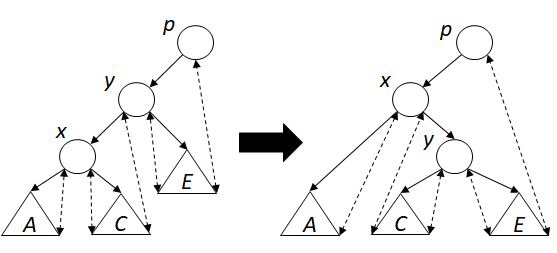}
\caption{\footnotesize Logical ordering tree.  Following a right rotation of $y$, a traversal about to navigate from
$y$ towards $A$ (before) would normally reach $C$ instead (after).  However, it can follow $\fprev$
pointers (dashed) to reach $A$ from $C$.}
\label{LO:Rotation}
\end{wrapfigure}%

Conceptually, the data structure consists of a doubly linked list sorted according to logical key order,
with a balanced binary tree superimposed on the list's nodes.  The binary tree structure is used to speed up
the set \code{insert}/\code{delete}/\code{contains} operations, but it is ultimately a key's presence or absence in
the list that determines an operation's result.  \Cref{Fi:LO-code} presents the code of the algorithm.
(The code is annotated with assertions in curly braces, which should be ignored at this point.)
For brevity, we omit the rebalancing logic, which decides when to perform rotations, as well as other auxiliary code
functions.

\begin{figure}
\centering
\begin{adjustbox}{width=\textwidth,center}
\begin{tabular}[t]{p{4.35cm}p{5.3cm}p{4cm}}
\begin{lstlisting}
type N 
  int key
  bool rem
  N left, right, parent
  Lock treeLock
  N succ, pred
  Lock succLock

N min$\leftarrow$new N($-\infty$);
N max$\leftarrow$new N($\infty$);
min$.$succ$\leftarrow$max
max$.$pred$\leftarrow$min
root$\leftarrow$max

private N tree-locate(int k)
  $\label{Ln:lo-locate-start}$x,y$\leftarrow$root
  while (y$\neq$null $\land$ y$.$key$\neq$k) $\label{Ln:lo-tree-search-start}$
    x$\leftarrow$y
    if (x$.$key<k)
      y$\leftarrow$x$.$right
    else
      y$\leftarrow$x$.$left  $\label{Ln:lo-tree-search-end}$
 $
  \color{assertion}
   {\begin{array}{l}
   \{\PSReachX{x} \\
   \ \ {}\land \PSReachX{y}\}
  \end{array}
  }
  \label{Ln:lo-RetLocate}
  $  
  $\label{Ln:lo-locate-end}$return (y=null? x : y)

bool contains(int k)
  x$\leftarrow$tree-locate(k)
  $\label{Ln:lo-contains-pred-start}$while x$.$key $>$ k
    $  
      \color{assertion}{\{ \PSReachX{x} \}}
      \label{Ln:lo-pred-search}
    $ 
    $\label{Ln:lo-contains-pred-end}$x$\leftarrow$x$.$pred
  $\label{Ln:lo-contains-succ-start}$while x$.$key $<$ k
    $  
      \color{assertion}{\{ \PSReachXK{x}{k} \}}
      \label{Ln:lo-succ-search}
    $ 
    $\label{Ln:lo-contains-succ-end}$x$\leftarrow$x$.$succ
  if (x$.$key $\neq$ k) 
   $  
    \color{assertion}{
    \begin{array}{l}
    \{ \Past{\QSReachXK{x}{k}} \land \QXK[>]{x}{k}  \}
    \end{array}
    }
    \label{Ln:lo-ContainsRetFalseBigger}
   $ 
    return false
  $\color{assertion}{\{ \PSReachXK{x}{k} \land \QXK{x}{k} \}  \label{Ln:lo-ContainsBeforeReadDEL}}$
  if (x$.$rem)$\label{Ln:lo-ContainsReadRem}$
   $  
      \color{assertion}{
      \begin{array}{l}
      \{ \Past{\QSReachXK{x}{k} \land \QXR{x}} \\
        \ \ {}\land \QXK{x}{k}  \}
      \end{array}
      }
      \label{Ln:lo-ContainsRetFalseRem}
    $ 
    return false
 $  
    \color{assertion}{
    \begin{array}{l}
    \{ \Past{\QSReachXK{x}{k} \land \neg \QXR{x}} \\
       \ \ {}\land \QXK{x}{k}  \}
    \end{array}
    }
    \label{Ln:lo-ContainsRetTrue}
  $ 
  return true
\end{lstlisting}
&
\begin{lstlisting}
bool delete(int k)
  x$\leftarrow$tree-locate(k)
  p$\leftarrow$(x$.$key$>$k ? x$.$pred : x)
  lock(p.succLock)
  s$\leftarrow$p$.$succ
  if k$\not\in$(p$.$key,s$.$key] $\lor$ p$.$rem
    restart
 $
  \color{assertion}{
  \begin{array}{l}
  \label{Ln:lo-delete-range}
  \{\QSReachXK{p}{k} \land k \in (p.\fkey,s.\fkey] \\
  \ \ {}\land \QXR[\neg]{p} \land p.\fnext=s \, \} \label{Ln:lo-delete-reachable-now}
  \end{array}}
  $  
  if s$.$key$\neq$k
    $
      \color{assertion}{\{\QSReachXK{s}{k} \land \QXK[>]{s}{k}\}} \label{Ln:lo-remove-ret-false}
    $  
    return false
  lock(s$.$succLock)
 $\color{assertion}{
  \begin{array}{l}
  \{\QSReachXK{s}{k} \land \QXK{s}{k} \land \QXR[\neg]{s}
  \}
  \end{array}
  }
  \label{Ln:lo-delete-write-rem}
  $
  $\label{Ln:lo-delete-write-rem-mark}$s$.$rem$\leftarrow$true
  removeFromTree(s)
  y$\leftarrow$s$.$succ
  y$.$pred$\leftarrow$p
 $
  \color{assertion}{
    \begin{array}{l}
    \{\QSReachXK{p}{k} \land p.\fnext=s \land s.\fnext=y   \\
      \ \ {}\land \QXK{s}{k} \land p.\fkey < k < y.\fkey \\
      \ \ {}\land \QXR[\neg]{p} \land \QXR{s}
    \}
    \end{array}
    }
  \label{Ln:lo-delete-write}
  $
  $\label{Ln:lo-delete-write-previous}$p$.$succ$\leftarrow$y
  return true

bool insert(int k)
  x$\leftarrow$tree-locate(k)
  p$\leftarrow$(x$.$key$>$k ? x$.$pred : x)
  lock(p.succLock)
  s$\leftarrow$p$.$succ
  if k$\not\in$(p$.$key,s$.$key] $\lor$ p$.$rem
    restart
 $
  \color{assertion}{
  \begin{array}{l}
    \label{Ln:lo-insert-range}
    \{\QSReachXK{p}{k} \land k \in (p.\fkey,s.\fkey] \\
      \ \ {}\land \QXR[\neg]{p} \land p.\fnext=s\,  \} \label{Ln:lo-insert-reachable-now}
  \end{array}}
  $ 
  if s$.$key$=$k
    $
      \color{assertion}{\{\QSReachXK{s}{k} \land \QXK{s}{k} \land \QXR[\neg]{s}\}} \label{Ln:lo-insert-ret-false} \label{Ln:lo-insert-ret-false-unmarked}
    $  
    return false
  n$\leftarrow$new N(k)
  $\label{Ln:lo-insert-write-new}$n$.$succ$\leftarrow$s
  n$.$pred$\leftarrow$p
  z$\leftarrow$chooseParent(p,s,n)
  n$.$parent$\leftarrow$z
 $\color{assertion}{
    \begin{array}{l}
    \{\QSReachXK{p}{k} \land \QXR[\neg]{p}   \\
      \ \ {}\land p.\fnext=s \land k \in (p.\fkey,s.\fkey) \\
      \ \ {}\land \QXK{n}{k}\land\QXR[\neg]{n} \land n.\fnext=s
    \}
    \end{array}
    }
  \label{Ln:lo-insert-write}
  $
  $\label{Ln:lo-insert-write-previous}$p$.$succ$\leftarrow$n
  s$.$pred$\leftarrow$n
  lock(z$.$treeLock)
  if (z$.$key$<$k)
    z$.$left$\leftarrow$n
  else
    z$.$right$\leftarrow$n
  return true
\end{lstlisting}
&
\begin{lstlisting}
private removeFromTree(n)  
  lock(n$.$treeLock)
  if (n$.$left $=$ $\NULL$)
    updateChild(n$.$parent,
                n,
                n$.$right)
    return
  if (n$.$right $=$ $\NULL$)
    updateChild(n$.$parent,
                n,
                n$.$left)
    return
  s$\leftarrow$n$.$succ
  lock(s$.$treeLock)
  // temporarily unlink s
  updateChild(s$.$parent,$\label{Ln:lo-remove2-unlink}$
              s,
              s$.$right)
  // s takes n's location
  s$.$left$\leftarrow$n$.$left $\label{Ln:lo-remove2-dup-start}$
  s$.$right$\leftarrow$n$.$right $\label{Ln:lo-remove2-dup-end}$
  n$.$left$.$parent$\leftarrow$s
  if (n$.$right $\neq$ $\NULL$)
    n$.$right$.$parent$\leftarrow$s
  $\label{Ln:lo-remove2-done}$updateChild(n$.$parent,n,s)

private updateChild(p,n,c)
  // pre: n locked
  // pre: n$.$parent $=$ p
  // pre: c = n's only child
  lock(p.treeLock)
  if (p$.$left $=$ n)
    p$.$left = c
  else
    p$.$right = c
  if (c $\neq$ $\NULL$)
    c$.$parent$\leftarrow$p

rotateRightLeft()
  p$\leftarrow$tree-locate($*$)
  lock(p.treeLock)
  y$\leftarrow$p$.$left
  if (y=null)
    return
  lock(y.treeLock)
  x$\leftarrow$y$.$left
  if(x=null)
    return
  lock(x.treeLock)
  $\label{Ln:lo-rotate-unlink-y}$p$.$left$\leftarrow$x
  p$.$parent$\leftarrow$y
  p$.$left$\leftarrow$x
  y$.$left$\leftarrow$x$.$right
  $\label{Ln:lo-rotate-link-y-back}$x$.$right$\leftarrow$y
\end{lstlisting}
\end{tabular}
\end{adjustbox}
\caption{\label{Fi:LO-code}
Logical-ordering tree~\cite{Drachsler:2014}. For brevity, \textbf{unlock} operations are omitted; a procedure releases all the locks it acquired when it terminates or \textbf{restart}s. $*$ denotes an arbitrary key.
}
\end{figure} 
The algorithm uses fine-grained locking. Each node has two locks, \code{treeLock} and \code{succLock}, for protecting
tree and list manipulations, respectively.  Every operation begins by traversing the tree, without acquiring any locks,
until reaching a leaf. A \code{contains} operation then traverses the linked list, searching for the target key,
and returns whether it was found.
\begin{changebar}
First it follows $\fprev$ pointers until it finds a node with a key not greater than the target, and if the target key was not yet found it continues to follow $\fnext$ pointers (see~\Cref{LO:traversal}). (The distinction between $\fprev$ and $\fnext$ pointers is important because they are not modified together atomically. As we shall see, the decisive view is of the list defined by $\fnext$ pointers.)
\end{changebar}
The \code{insert} and \code{delete} operations acquire locks, verify that they
have landed at the correct location in the list, and then perform their respective operation on both the list and
the tree. The \emph{update} operations, \code{insert} and \code{delete}, are \emph{successful} if they insert/delete
the key from the data structure, and \emph{failed} otherwise.
\begin{figure}
\includegraphics[scale=0.4]{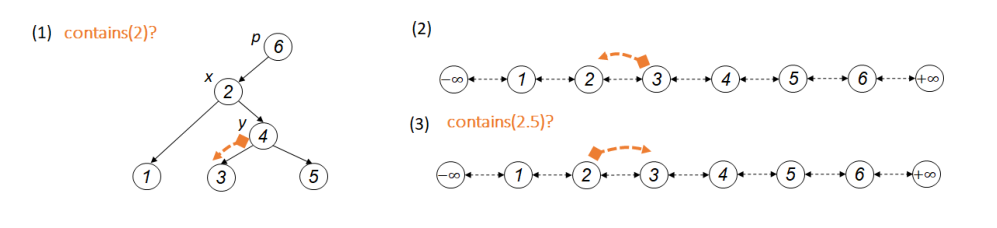}
\caption{\footnotesize Traversals in the LO tree. A traversal looking for key $2$ has reached node $y$ due to a concurrent rotation in the tree (see~\Cref{LO:Rotation}). In (1), it continues to perform a binary search in the tree and does not find $2$ (\crefrange{Ln:lo-tree-search-start}{Ln:lo-tree-search-end} in~\Cref{Fi:LO-code}). Nevertheless, in (2), the traversal continues by reading $\fprev$ pointers and finds the key, allowing it to return $\code{contains}(2)=\true$ (\cref{Ln:lo-pred-search}). Another traversal, this time looking for key $2.5$, encountering the same scenario in (1)--(2) performs an extra step of reading $\fnext$ pointers until it reaches a node with a larger key, allowing it to return $\code{contains}(2.5)=\false$ (\cref{Ln:lo-succ-search}).}
\label{LO:traversal}
\end{figure}%

Deletions are performed in two steps. The node is first \emph{logically} deleted, removing its key from the set
represented by the tree, by setting its boolean \code{rem} field. The node is then \emph{physically} removed from
the tree and list. Physical removal of a leaf or a node with a single child simply splices the node from the
tree.
\begin{figure}
\includegraphics[scale=0.635]{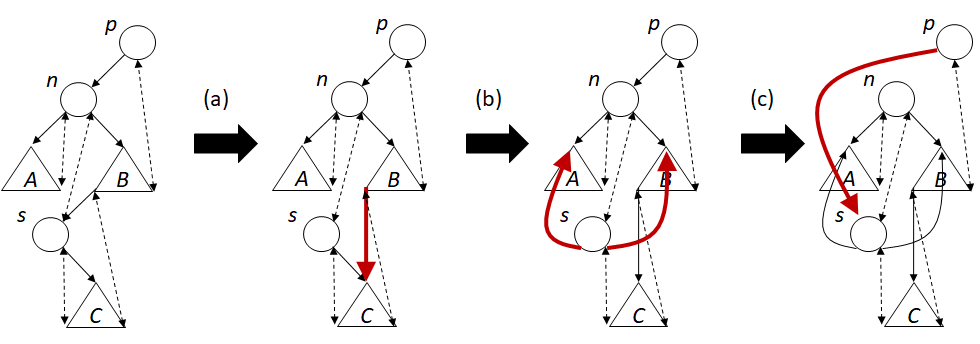}
\caption{\footnotesize Removing a node $\code{n}$ with with two children from the tree structure of the LO tree (\code{removeFromTree} in~\Cref{Fi:LO-code}). 
$\code{s}$ is the successor of $\code{n}$, found using the list layout. (a) $\code{s}$ is temporarily removed from the tree structure (\cref{Ln:lo-remove2-unlink}). (b) the children of $\code{n}$ are copied to $\code{s}$ (\crefrange{Ln:lo-remove2-dup-start}{Ln:lo-remove2-dup-end}). (c) $\code{n}$'s parent is modified to point to $\code{s}$ instead; thus $\code{s}$ takes the location of $\code{n}$ in the tree (\cref{Ln:lo-remove2-done}).
Note that $\fprev,\fnext$ are not modified yet, and are thus inconsistent with the tree layout (the list layout is updated afterwards). 
The updates to the $\fparent$ field are not presented.}
\label{Fi:LO:remove2c}
\end{figure} 
 Physical removal of a node with two children is more subtle (see~\Cref{Fi:LO:remove2c}). It is performed by
replacing the deleted node with its successor, which is obtained from the list. Similarly to a rotation, moving the
successor up the tree in this way%
\footnote{In the tree, the successor is the leftmost leaf of the node's right subtree.}
can cause concurrent tree traversals searching for the successor to miss it \begin{changebar}(in~\Cref{Fi:LO:remove2c}, \code{s} is unreachable in the tree after modification (a) and before (c))\end{changebar}. For this
reason, other concurrent BSTs~\cite{Arbel:2014,Bronson:2010,PFL16:CFTree} implement internal node deletion differently.
The LO tree solves the problem by relying on the underlying list structure for correctness.
\begin{changebar}
In insertion, \code{chooseParent} returns the location in the tree to which the new node should be linked.
We omit its code,
because (as we shall see in the proof below), the correctness of tree operations stems from the list structure. A node's
tree location only determines the efficiency of a traversal locating the node.
\end{changebar}

Overall, successful update operations mutate both the tree and list structure,
which involves writing to multiple memory locations. Concurrent traversals can observe a subset of these writes,
effectively observing the data structure in mid-update.  Reasoning about all possible interleavings of writes
and traversal reads is incredibly hard.  Indeed, we find that the LO tree's original linearizability proof~\cite{Drachsler:2014} is flawed.%
\footnote{This has been confirmed with Drachsler et al.}
The proof claims that an \code{insert}($k$) takes effect when $k$'s node gets pointed to by its predecessor's $\fnext$
pointer.
However, a \code{contains} searching for $k$ could find it before this update by following $\fprev$ pointers.  Consequently, the code we present in~\cref{Fi:LO-code} uses a different order of pointer
updates than the original code~\cite{Drachsler:2014}: in \code{insert}, we \emph{first} make the key reachable by $\fnext$ links, and only then add it to the tree and $\fprev$ links; in \code{delete}, we remove the node \emph{last} from the successors list, after it is unreachable by tree and $\fprev$ links.\footnote{
	Earlier versions of the current paper have shown in~\Cref{Fi:LO-code} an incorrect order of insertions in \code{insert}, despite referring to the correct order in the proof in~\Cref{sec:lo-traversal-full}. This was pointed out by~\citet{DBLP:journals/pacmpl/0001W023}.
}

\subsection{Linearizability of the Logical Ordering Tree}
\label{sec:lo-linearizability}
In this section, we give an overview of a standard linearizability proof for the LO tree.
We show that the proof boils down to reasoning about traversal correctness, and explain why proving traversal
correctness is non-trivial.

The proof uses an abstraction function $\repfunc$ mapping each concrete memory state to the (abstract, mathematical) set it represents, and showing that every operation ``takes effect'' at some point during its execution (see~\Cref{sec:linearizability-background}.)
To define the abstraction function $\repfunc$, we conceptually break the LO tree's doubly linked list into
the \emph{successor} and \emph{predecessor} lists, induced by following $\fnext$ pointers from the sentinel
node \code{min} (denoted $\succroot$) and by following $\fprev$ pointers from the sentinel node \code{max}
(denoted $\predroot$), respectively.  We define $\repfunc$ based on reachability in the successors list.
$\repfunc$ maps a state $\state$ to the set of the keys of the nodes on the successor list that are not logically
deleted.  Technically, for every key $k$, we define a state-predicate over memory locations $\ell$, $\QSReachXK{\ell}{k}$.
This predicate holds in state $\state$ if, when following the $\fnext$ pointers from $\succroot$ in $\state$ and stopping when $k$
or a greater key is found, we encounter $\ell$. \yotamsmallnew{added:}We then say that $\ell$ is \emph{$k$-reachable}.
The abstraction function is then defined as follows:
\[
\begin{multlined}[t]
\repfunc(\state) = \{ k \in \mathbb{N} \mid \state \models \exists x.\,\QSReachXK{x}{k} \land x.key=k \land \QXR[\neg]{x}\}, \\
\mbox{ where $\state \models P$ means that $P$ is true in $\state$.}
\end{multlined}
\]
The challenge now is
to identify when an operation's execution takes effect on the abstract state %
by establishing properties of the data structure and the algorithm.
These are displayed as~\emph{assertions} annotating the code in~\Cref{Fi:LO-code}. %
The assertion notation
should be interpreted as follows.
An assertion $\{P\}$ means that in every execution, $P$ holds in any state in which the next line of code executes,
which can be thought of as ``now'' (with respect to the executing operation's perspective).
An assertion $\{\Past{P}\}$ says that $P$ was true at some point between the invocation of the
operation and ``now.'' Note that assertions use both $\QSReachXK{\ell}{k}$ as well as another reachability predicate,
$\QSReachX{\ell}$, which holds in a state whenever $\ell$ is reachable from $\succroot$ by following $\fnext$ pointers (irrespective of any key).\footnote{
	For brevity, we omit from the assertions facts about locks, which in this algorithm can always be inferred from the scope.
}

Assuming that these assertions hold, it is a rudimentary exercise to show that any operation takes effect on the abstract set (see~\Cref{sec:linearizability-background}) during some point of its execution; \begin{changebar}see~\refappendix{sec:lo-appendix} for an elaboration\end{changebar}. The important point to notice is that reachability assertions are inherent to showing the necessary $k \in \repfunc(\sigma)$ or $k \not\in \repfunc(\sigma)$ during the operation's execution. For example, that a node with key $k$ was previously reachable when \code{contains} returns true (\cref{Ln:lo-ContainsRetTrue}) is crucial for showing that at some point $\sigma$ during the operation's execution indeed $k \in \repfunc(\sigma)$.

\para{Proving the assertions}
The proof's main goal is thus to prove the assertions in~\Cref{Fi:LO-code}, which would conclude the linearizablity proof. Which assertions are hard to prove?
The assertions %
that do not concern traversals and $\{\Past{P}\}$ properties
are straightforward, and rely on %
simple inductive invariants about the data structure, properties such as the immutability of keys, and
reasoning about interleavings of the lock-protected critical sections. \begin{changebar}(See~\refappendix{sec:lo-appendix} for an elaboration.)\end{changebar}

The crux of the linearizability proof, therefore, is to prove the assertions that concern the
unsynchronized traversals (\cref{Ln:lo-RetLocate,Ln:lo-pred-search,Ln:lo-succ-search,Ln:lo-ContainsRetFalseBigger,Ln:lo-ContainsRetTrue,Ln:lo-ContainsBeforeReadDEL,Ln:lo-ContainsRetFalseRem}).
These assertions state that nodes were reachable \emph{at some point in time} during an operation's execution. We formalize the problem of proving such properties as that of establishing \emph{traversal correctness}.

\subsection{The Need for Traversal Correctness Assertions}
\label{lo:reachable-only-in-past}
The focus of our work is $\{\Past{P}\}$ properties, for predicates $P$ that have to do with reachability.
We use such properties to capture traversal correctness.
An example is the assertion $\PSReachXK{x}{k}$ in~\cref{Ln:lo-succ-search}, which states that when the traversal on successor nodes in~lines~\ref{Ln:lo-contains-succ-start}--\ref{Ln:lo-contains-succ-end} executes, every node the traversal encounters has been $k$-reachable at some point between the method's invocation and now.

Traversal correctness guarantees the reachability of encountered nodes at \emph{some} time in the past.  (We formally define traversals and traversal correctness in~\Cref{sec:traversals}.)  Such a property seems less intuitive than the more natural claim that these nodes are reachable \emph{when they are encountered}. That easier claim would indeed have been nice to have, only that it does not hold, and $\QSReachXK{x}{k}$ would be an incorrect assertion in~\cref{Ln:lo-succ-search}---a traversal can arrive at a node \emph{after it is no longer reachable}.  For example, a traversal can arrive at node $x'$, but before it continues to $x=x'.\fnext$, $x'$ and then $x$ are removed.  Thus, when the traversal reads $x'.\fnext$ and arrives at $x$, the node $x$ is no longer reachable.

Consequently, the crux of the linearizability argument follows from traversal correctness properties.  Unfortunately, traversal correctness is fundamentally harder than proving the other assertions in the code, as explained next.

\subsection{The Challenge of Proving Traversal Correctness}
Proving traversal correctness requires showing that a node was reachable %
at \emph{some point in time} during the traversal's execution.  While it is immediate that if $x$ is reachable
now and $x$ points to $y$ now, then $y$ is reachable now, analogous reasoning is not straightforward for
reachability ``in the past.''
Consider, for example, a list traversal that reads a pointer $x.\fnext$ which points to node $y$.  Knowing that
 $\Past{\QSReachX{x}}$ does not directly imply the reachability of $y$---neither now nor in the past---because
 $\QSReachX{x}$ might \emph{not} hold when $x.\fnext$ is read.
Imagine what could happen in between: For example, as explained above, $x$ could be removed, and subsequently $y$ could be removed, and still $x.\fnext = y$, even though $y$ is no longer reachable. Although there is a link $x.\fnext = y$ and it is easy to see (from the assertions) that $y$ was reachable when this link was written, this is not evidence enough for $\Past{\QSReachX{y}}$, because this write could have taken place before the beginning of the traversal (and before the invocation of the current method).
Another tricky scenario to consider is that although $y$ can be removed, its key could be inserted elsewhere in a new node $z$, and the traversal may not be aware of this and miss $z$.
It is therefore possible for traversals to reach a key that has already been removed, ``miss'' insertions, etc.
Furthermore, the path that the traversal follows is built of pointers from many points in time, and may never have existed in memory in its entirety.
Surprisingly, it turns out that deciding existence/absence of a key based on such a traversal is correct.
That is, for a list traversal that starts at $\succroot$, it can indeed be established that $\Past{\QSReachX{y}}$ holds
in this example, but this involves %
intricate reasoning about how interfering writes interleave with the traversal's reads. Simplifying this has been the goal of previous works~\cite{PODC10:Hindsight,FeldmanE0RS18}. %

The traversal in the LO tree is even more complex.
While it is the reachability in the doubly linked list that determines the
result of a traversal, its first part, {\tt tree-locate}, traverses the nodes
over the binary tree links.
The traversal over the list \emph{does not start from the head of the list}, but from where the tree traversal landed, so it is necessary to prove properties of the tree-traversal as well.
Reasoning about the tree traversal is even more challenging than the traversal over the list, for two reasons: %
\begin{enumerate}
		\item Having to reason about lists while traversing a
tree makes ``off-the-shelf'' approaches not suitable for proving the
reachability assertions during traversals (see~\Cref{sec:related}).
Intuitively, {\tt tree-locate} accesses nodes in the list in almost a random-access way, so it is unclear why it cannot go wrong.
		\item Moreover, examining the
tree traversal is very challenging due to complex interference patterns: in between this traversal's reads, in-place tree rotations can take place, nodes can be unlinked and linked back, etc. Previous approaches fall short, and cannot be applied to liberate the user from considering the interleavings of writes in between the traversal's reads (see~\Cref{sec:related}).
\end{enumerate}

\subsection{Our Framework: Proving Traversal Correctness}
\label{sec:overview-framework-successor-links-correctness}
Our proof framework provides a \emph{general} method for proving traversal correctness %
while reasoning only about the effect writes have on memory, \emph{without} resorting to complex concurrent reasoning
about how reads in the traversal interleave with interfering writes (in the like of the corner cases above).
\begin{changebar}
\yotamsmallnew{postponed overview to the end, here just mention that we also do everything later}
In the remainder of this section we illustrate the ideas behind our framework by (informally) proving one of the traversal correctness assertions, $\PSReachXK{x}{k}$ in~\cref{Ln:lo-succ-search} for the segment of the traversal over successor links
(\yotamsmallnew{added:}we prove all the remaining traversal correctness assertions in~\Cref{sec:lo-traversal-full}, after formally presenting the framework).

Consider the traversal in~\crefrange{Ln:lo-contains-succ-start}{Ln:lo-contains-succ-end}. For us, a traversal is simply a sequence of read operations of locations $\ell_0,\ell_1,\ldots$ (\Cref{sec:traversals})\sharonnew{shouldn't this be explained much earlier?}\yotamsmallnew{I don't know. Only here we rely on the formalization of this, sequences of locations etc., no?}, and here this sequence traverses successor pointers until it reaches a node with a key greater or equal than $k$.

Suppose that we have already established the assertion $\PSReachXK{x}{k}$ before this traversal, that is, at the entry to the loop in~\cref{Ln:lo-contains-succ-start}. Our goal is to prove that, in spite of possible interference, $\PSReachXK{x}{k}$ holds also for the new $x$'s that the traversal reaches by following successor links.
Our framework achieves such a proof by showing that two properties hold.
The first, \emph{single-step compatibility}, connects the reachability predicate $\QSReachXK{\cdot}{k}$ to the way the traversal navigates from one node to the next. %
The second, \emph{forepassed}, constrains the effect over time of interfering writes on the reachability predicate.

\para{(1) Single-step compatibility}
The traversal chooses the next location to read based on the last read location and the key; for example, the traversal reads $\ell_i = o.\fkey$ and decides to proceed by reading $\ell_{i+1}=o.\fleft$ according to the binary search. This is called a \emph{step} of the traversal. The next step in this case is to read $\ell_{i+2}=o'.\fkey$ where $o'$ is the object to which $o.\fleft$ points to (dereferencing the pointer), from which the traversal proceeds by taking steps in a similar fashion.

The requirement of \emph{single-step compatibility} (\Cref{def:local-path-extension}) focuses on a single step of the traversal at each time. Consider the traversal advancing from $\ell_i$ to $\ell_{i+1}$ by reading the single value at the location $\ell_i$ from the current memory state $\sigma_t$.
Then single-step compatibility requires that if $\sigma_t \models \QSReachXK{\ell_i}{k}$, then also $\sigma_t \models \QSReachXK{\ell_{i+1}}{k}$ holds.
Single-step compatibility obviously holds for the traversal over successor links and this reachability predicate (see~\Cref{ex:lo-local-path-extension}).
Note that %
both $\QSReachXK{\ell_i}{k}$ and $\QSReachXK{\ell_{i+1}}{k}$ here are evaluated in the same memory state $\sigma_t$ (without interference).
Importantly, this means that establishing this condition relies on purely \emph{sequential} reasoning: the scope of this condition is a single read operation, and interference is irrelevant.

\para{(2) Forepassed interference}
This condition tracks writes that \emph{reduce k-reachability}: the $k$-reachability of a location $\ell$ is \emph{reduced} by the write $w$ if before $w$ it holds that $\QSReachXK{\ell}{k}$ but not after $w$.
For example, removing a node (by modifying the pointer of its parent) reduces its reachability. Intuitively, such an interfering write is ``dangerous'' because a traversal can reach $\ell$ and be unaware that $\ell$ is no longer $k$-reachable.
The \emph{forepassed} condition (\Cref{def:black-condition}) requires that a location $\ell$ whose reachability is reduced by $w$ at time $t$ either
\begin{itemize}
	\item cannot be later modified (we call this \emph{strong forepassed}), or,
	\item otherwise, if $\ell$ is modified by a later write $w'$, %
	then $w'$ writes a value that points to a location $\ell'$ that is known to have been $k$-reachable at some intermediate moment, %
	between the time immediately before $w$ was performed and the time immediately after $w'$ was performed. %
\end{itemize}

In the traversal over successor links in the LO tree, %
whenever the reachability of a node is reduced it is first \emph{marked}, inhibiting later writes to this location (see~\Cref{ex:lo-black-condition-holds}), which guarantees the forepassed condition.

The insight behind our approach and the forepassed condition is that accessing a memory location $\ell$ whose reachability has been reduced is not itself problematic; $\ell$ still \emph{has been} $k$-reachable.
The challenge is to prove that locations reached through $\ell$ have also been $k$-reachable.
The main idea of the forepassed condition is that to achieve that, we must limit the ways $\ell$ can be modified after its reachability is reduced.
Consider the next location the traversal visits, $\ell'$, which is pointed to by $\ell$. 
In any point in time where $\ell$ was $k$-reachable and contained the same value, $\ell'$ was also $k$-reachable. 
Thus, if $\ell$ was not later modified prior to reading its value and visiting $\ell'$ (this is the case of strong forepassed), then $\ell'$ has also been $k$-reachable. 
The forepassed condition still allows some writes to $\ell$ after its reachability was reduced, as long as they retain the property that the next location in the traversal $\ell'$---which is dictated by the values these writes put in $\ell$---also has been $k$-reachable.
Either way, the forepassed condition guarantees that the traversal continues to locations that have been $k$-reachable, and the traversal is not ``led astray''.

In almost all our examples, the interference satisfies the \emph{strong forepassed} condition, which is simpler to reason about, but the more general condition is required e.g.\ for an implementation of backtracking (see~\Cref{sec:cf-short}).

\para{Deducing traversal correctness}
The core of our framework is the theorem that if the traversal is single-step compatible and writes satisfy forepassed interference, both w.r.t.\ the reachability predicate, then the traversal is correct---every location it reaches has been reachable at some point (\Cref{thm:main-thm}). Thus,
\yotamsmallnew{rephrased} in the traversal over successor links of the LO tree,
we deduce that every $x$ the traversal reaches satisfies $\PSReachXK{x}{k}$, finally proving our much sought traversal correctness assertion.
\adamnew{suggest to tone down (and not by saying that there is just a bit of rejoicing)}\yotamsmallnew{I hope you're happy, I hope you're happy now, I hope you're happy how you've Hurt your cause forever}

\begin{remark}
\label{remark:choose-reach}
Traversal correctness, single-step compatibility, and the forepassed condition all depend on the choice of reachability predicate. In the example we have considered here, the reachability predicate was $\QSReachXK{\cdot}{k}$, and it arose directly from the definition of the abstraction function, which is usually the case.
Occasionally, choosing the right reachability predicate demands more care.
In the LO tree, the traversal over the tree and the predecessor links (lines~\ref{Ln:lo-locate-start}--\ref{Ln:lo-locate-end},\ref{Ln:lo-contains-pred-start}--\ref{Ln:lo-contains-pred-end}) is correct w.r.t.\ $\QSReachX{\cdot}$ irrespective of the key (which makes sense because it performs ``random accesses'' to the successors list). This is discussed in~\Cref{sec:lo-traversal-full}.
A more intricate scenario is the Citrus tree~\cite{Arbel:2014}, discussed in~\Cref{sec:citrus-short}. There, traversal correctness requires a weakening of the reachability predicate that captures (using ghost state) the potential of paths to go off track in certain ways.
\end{remark}
\end{changebar}

\para{Scope and limitations}
\begin{changebar}
We are not aware of algorithms in our domain where our proof argument is inherently inapplicable, although such examples are possible.
The Citrus tree (\Cref{sec:citrus-short}) comes close, in the sense that our framework does not apply with a straightforward choice of the reachability predicate (in fact, the traversal is not correct in the sense of~\Cref{def:traversal-correctness}), but we are nonetheless able to apply our technique using a modified reachability predicate and ghost code.
\end{changebar}

\begin{changebar}
In this paper we assume sequential consistency (SC), following most concurrent data structures and many papers on their verification. In many relaxed memory models, data-race free programs have SC semantics, making our techniques applicable (in C/C++, for example, the shared node fields would be accessed with SC atomics). Proving the correctness of algorithms that exploit weaker consistency models is an interesting future direction.
\end{changebar}

\begin{changebar}
\para{Outline}
We formally define traversals and traversal correctness in~\Cref{sec:traversals}, and develop the framework's conditions and main theorem in~\Cref{sec:proving-traversals}.
In~\Cref{sec:lo-traversal-full}, we apply it to prove all the traversal correctness assertions of the LO proof, thereby completing the linearizability proof here.
In~\Cref{sec:circular}, we discuss how traversal correctness assertions can rely on the very same assertions that they help prove, in an inductive manner, making such proofs simple.
We apply the framework to additional challenging examples in~\Cref{sec:additional}.
Related work is discussed in~\Cref{sec:related}, and~\Cref{sec:conclusion} concludes.
\end{changebar}

\section{Traversals and Their Correctness Criterion}
\label{sec:traversals}
Our framework targets traversals. In this section we describe our model of traversals as a sequence of reads determined by a sequence of local steps (\Cref{sec:sub-traversal}),
\begin{changebar}
a view which is essential to phrase our single-step compatibility condition below (in~\Cref{sec:local-path-extension}).
\end{changebar}
We then formally define traversal correctness w.r.t.\ a reachability predicate (\Cref{sec:traversal-correctness-formal})
\begin{changebar}
which is the central proof goal of our framework.\footnote{A similar notion was implicit in~\cite{FeldmanE0RS18}.}
\end{changebar}
As a running example, we use \yotamsmallnew{added:}the same traversal whose correctness we described informally in~\Cref{sec:overview-framework-successor-links-correctness}: the traversal in \crefrange{Ln:lo-contains-succ-start}{Ln:lo-contains-succ-end} of \code{contains} of~\Cref{Fi:LO-code} %
and its correctness with respect to the reachability predicate $\QSReachXK{x}{k}$.
We begin with some preliminary definitions.

\subsection{States, Locations, Executions, and Writes} \label{sec:definitions}
A \emph{state}, %
denoted $\sigma$, is a mapping from memory locations to values.
We use an indexing of memory locations by pairs, $(o,\mathbf{f})$,
where $o$ is an object identifier, and $\mathbf{f}$ is a field name. The value in a location can be another location (a ``pointer''---e.g.\ to $\oloc{o}{\fkey}$). %
A \emph{write} is a pair $(\ell,v)$ of a location $\ell$ and the value $v$ being written to it.
We use discrete timestamps $\ldots, t-1, t,t+1,\ldots$ to model the order in which writes occur.
An \emph{execution} is a sequence of (atomic) %
writes (performed by different threads) at increasing timestamps %
performed by the algorithm---this corresponds to recording just the write operations in a run of the algorithm.
Given an execution, we denote by $\sigma_t$ the state of the algorithm at time $t$,
and by $\sigma_t(\ell)$ the value in location $\ell$ in $\sigma_t$.
States are modified by writes: if a write $(\ell,v)$ occurs at time $t$, then $\sigma_{t}(\ell') = \sigma_{t-1}(\ell')$ for every $\ell' \neq \ell$, and $\sigma_{t}(\ell) = v$
(that is, $\sigma_t$ is the state after the write is performed).
When we consider an execution in the \emph{timespan} $[\texbegin, \texend]$, then the first state in the execution is $\sigma_{\texbegin}$, and the execution consists of writes with timestamps in $(\texbegin, \texend]$. (The first timestamp may be an arbitrary point in the algorithm's run, not necessarily the beginning of any operation.)
A \emph{read} is performed to a memory location $\ell$ from a memory state $\sigma_t$, observing the value $\sigma_t(\ell)$. For our purposes, we shall not need to model the exact time when a read occurs, only the state from which it reads the value (when ordering the reads and writes together the read occurs after write $t$ and before $t+1$).

\begin{changebar}
\begin{remark}
\label{remark:code-vs-reads-writes}
In our formalization, reads and writes occur in specific memory states, to/from specific memory locations. There is a subtle gap between this and the \emph{code} of the algorithm (such as the code in~\Cref{Fi:LO-code}); it is necessary to translate the read and write program instructions to the actual read and write operations performed when the code executes, which is the level of abstraction our formalization uses. Bridging such a gap in a formal proof is usually the role of a program logic; in this paper, when we apply the framework to prove algorithms presented in code, this connection is straightforward.
\end{remark}
\artemnew{I do not understand how what is in this remark is a role of a
program logic. Don't executions come from running a program? (For instance,
histories for linearizability are not generated with a program logic.) I see
this as the role of the semantic model for a programming language.

This subsection is introducing such a model though (?).}
\end{changebar}

\subsection{Traversals}
\label{sec:sub-traversal}
To %
formally define a traversal, we use
a relation $\extend{p}(\ell,v,\ell')$, which encodes that after reading location $\ell$ and seeing value $v$, the location $\ell'$ is the next to be read. %
This relation captures the sequence of reads a traversal (e.g.\  \crefrange{Ln:lo-contains-succ-start}{Ln:lo-contains-succ-end} of \code{contains} in the~\Cref{Fi:LO-code}) performs.
The specialization by $p$ denotes the dependence of the traversal's logic on certain parameters, such as a key in our running example. %
For a given $p$, the relation $\extend{p}$ only depends on $\ell,v,\ell'$. %
The intuition is that the next location $\ell'$ to be read is chosen based on the last location read $\ell$ and the value read $v$ (but nothing else). The definition of the relation reflects the code that implements the traversal.

Given an execution, a \emph{traversal} is
a sequence $(\ell_1,t_1),\ldots,(\ell_n,t_n),\ell_{n+1}$, where $\ell_i$'s are locations, $t_i$'s are timestamps,
and every consecutive pair satisfies $\extend{p}$, namely, $\forall i=1,\ldots,n. \ \extend{p}(\ell_i, v_i, \ell_{i+1})$ where $v_i = \sigma_{t_i}(\ell_i)$.
The pair $(\ell_i,t_i)$ indicates that the traversal reads $\ell_i$ from the state $\sigma_{t_i}$ (observing the value $\sigma_{t_i}(\ell_i)$).
When
the traversal performs the read of location $\ell_i$, we say that
the traversal \emph{reaches} $\ell_{i+1}$. Note that at this point,
$\ell_{i+1}$ itself is not (yet) read---as an illustration, $\ell_{i+1}$ may be reached by ``following the pointer'' in $\ell_i$, which amounts to reading $\ell_i$.
\begin{example}\label{ex:traversal}
The \code{contains} operation of the LO tree visits a sequence of nodes that can be split into two sequences:
\begin{inparaenum}[(1)]
\item from the beginning of \code{contains} until
\cref{Ln:lo-contains-pred-end}, the operation visits nodes by following {\tt
left}, {\tt right} and {\tt pred} links;
\item afterwards, at
\crefrange{Ln:lo-contains-succ-start}{Ln:lo-contains-succ-end}, the operation
visits nodes by following {\tt succ} links.
\end{inparaenum}
We use the latter in our illustrations %
of the framework
throughout \Cref{sec:traversals} and~\Cref{sec:proving-traversals}, and give full details for the LO tree in~\Cref{sec:lo-traversal-full}.

Consider the traversals over successor links that a {\tt contains} operation
performs searching for a key $k$ at
\crefrange{Ln:lo-contains-succ-start}{Ln:lo-contains-succ-end} in
\Cref{Fi:LO-code}. To analyze them, we instantiate $\extend{p}$ with a relation $\extend{k}$,
which is parametrized by the key $k$ and is true only for the following cases (for every objects $o$ and $o'$, and value $m$):
\begin{align*}
	&\extend{k}(\oloc{o}{\fkey},m,\oloc{o}{\fnext}) &\quad \mbox{if $m<k$}
	\\
	&\extend{k}(\oloc{o}{\fnext},\oloc{o'}{\fkey},\oloc{o'}{\fkey}) & %
\end{align*}
Informally, this definition states that \begin{inparaenum}[(i)]
	\item from a $\fkey$ field of an object, the traversal is extended to the $\fnext$ field of the same object, in case the value in the $\fkey$ field is less than $k$, and
	\item from a $\fnext$ field the traversal is extended to the location holding the $\fkey$ field of the object to which $\fnext$ points (following the pointer, in short).
\end{inparaenum}
(In the assertions in~\Cref{Fi:LO-code}, when we write $\QSReachXK{x}{k}$, reachability to the object $x$ is a shorthand to reachability to $(x,\fkey)$.)
It is immediate that $\extend{k}$ correctly summarizes the code performing the traversal in~\crefrange{Ln:lo-contains-succ-start}{Ln:lo-contains-succ-end} of \code{contains},
namely, in the sequence of reads performed by this operation %
executing this code section, every consecutive pair of locations satisfies $\extend{k}$ with the corresponding value read (see also~\Cref{remark:code-vs-reads-writes}).
\end{example}

\subsection{Traversal Correctness}
\label{sec:traversal-correctness-formal}
Roughly, traversal correctness requires that
every location reached by the traversal has been ``reachable''
within a certain preceding timespan.
The notion of reachability and the timespan are formalized next.

\para{Reachability}
Reachability
is defined by a reachability predicate,
chosen to be useful in the overall correctness argument, such as
$\QSReachXK{x}{k}$ in~\Cref{sec:lo-linearizability} \begin{changebar}(see also~\Cref{remark:choose-reach})\end{changebar}. Formally, a
\emph{reachability predicate} is a unary state-predicate over locations that can
be parameterized, e.g.\ by a (fixed) root and by a key of interest.
We denote this predicate $\QReachAbs{\cdot}{k}$. For
a location $\ell$, we use $\models_t \QReachAbs{\ell}{k}$ to denote that the
predicate $\QReachAbs{\ell}{k}$ holds in the state $\sigma_t$ (the state at time
$t$). We also say that $\ell$ satisfies the reachability predicate at time $t$.

\begin{example}\label{ex:reach-pred}
\label{ex:k-reach-def}
In the running example, the reachability predicate of interest for the traversal in \crefrange{Ln:lo-contains-succ-start}{Ln:lo-contains-succ-end} %
searching for some key $k$
is $\QSReachXK{x}{k}$, called $k$-reachability. %
The $k$-reachability predicate is defined
by the existence of a sequence of locations that follows $\extend{k}$:
$\QSReachXK{x}{k}$ holds in state $\sigma$ if there is a sequence of locations $\ell_0,\ldots,\ell_n$ s.t.\ $\ell_0 = \oloc{\succroot}{\fkey}$,
$\ell_n = x$, and $\forall i < n. \ \extend{k}(\ell_i,\sigma(\ell_i),\ell_{i+1})$.
\end{example}

Our framework can accommodate versatile reachability predicates, including reachability and $k$-reachability in a list (\Cref{sec:lo-traversal-full}), $k$-reachability by binary search in a tree (\Cref{sec:cf-short}), and even sophisticated reachability using ghost state (\Cref{sec:citrus-short}). %

\para{Base time}
In addition to the reachability predicate, traversal correctness is also relative
to a \emph{base} time $\texbegin \leq t_1$ within the timespan of the current operation (that is, the state $\sigma_{\texbegin}$ was present concurrently with the operation),\sharonnew{note time issue. timespan was defined, but not timespan of operation. Does it start at the last write before the operation?}
such that the first location $\ell_1$ is reachable at base time $\texbegin$.
The base time, similarly to the reachability predicate, is chosen as part of the proof. %
Most often, $\texbegin = t_1$ and $\ell_1$ is the root that is always reachable, but a
base $\texbegin \leq t_1$ is needed e.g.\ in the running example (see~\Cref{sec:lo-traversal-full}). %
Below, for convenience, we consider executions whose timespan begins at the base time $\texbegin$.

\para{Traversal correctness}
Given the reachability predicate and the base, \emph{traversal correctness} requires that every location $\ell_{i+1}$ reached by the traversal at time $t_i$
as defined in~\Cref{sec:sub-traversal} has satisfied the reachability predicate at some point (state) %
in the execution between $\texbegin$ and $t_i$ (inclusive).
Let $\Pastlimit{\QReachAbs{\ell}{k}}{t}{t'}$ be a shorthand for $\exists t''. \ t \leq t'' \leq t' \land  \models_{t''} \QReachAbs{\ell}{k}$. Then we can define traversal correctness as follows.

\begin{definition}\label{def:traversal-correctness}
A traversal $(\ell_1,t_1),\ldots,(\ell_n,t_n),\ell_{n+1}$ is \emph{correct} w.r.t.\ $\QReachAbs{\cdot}{}$ and base $\texbegin$ if $$\Pastlimit{\QReachAbs{\ell_{i+1}}{}}{\texbegin}{t_i}$$ holds
for every $i=0,\ldots,n$, with $t_0 = \texbegin$.
\end{definition}
Note that traversal correctness requires that $\ell_1$ is reachable at the base time $\texbegin$.

In our applications of the framework to prove different algorithms we always use timestamps that are concurrent with the current operation, and so use $\PastSymb$ without time bounds to indicate that $\texbegin$ is the beginning of the operation and $\texend$ is the current time. Most often we are interested in the reachability of the last location reached in the traversal, and deduce the assertion $\Past{\QReachAbs{\ell_{i+1}}{k}}$ in the code performing the traversal.
The goal of our framework is to prove traversal correctness using simple concurrent reasoning, as we describe next.

\section{The Framework: Proving the Correctness of Traversals}
\label{sec:proving-traversals}
In this section we describe how our framework proves traversal correctness.
We first explain how the reachability predicate needs to be tied to the traversal itself via single-step compatibility (\Cref{sec:local-path-extension}), and explain the forepassed condition about writes (\Cref{sec:condition-on-inteference}) that guarantees traversal correctness in spite of interference (\Cref{sec:main-theorem}). Lastly, we extend the framework to deduce reachability together with a property of a single field (\Cref{sec:reachability-with-field}).%

Throughout this section, we fix an execution of the algorithm in timespan $[\texbegin,\texend]$.

\subsection{Single-Step Compatibility}
\label{sec:local-path-extension}

Our framework is applicable to prove traversal correctness w.r.t.\ reachability predicates that are compatible with the $\extend{p}$ relation underlying the traversal in the following way:

\begin{definition}\label{def:local-path-extension}%
We say that a reachability predicate $\QReachAbs{\cdot}{k}$ is {\em single-step compatible}
with %
an $\extend{p}$ relation if for every state $\sigma$ %
in the execution and every pair of locations $\ell,\ell'$ it holds that
$\sigma \models \QReachAbs{\ell}{k} \land \extend{p}(\ell,\sigma(\ell),\ell') \implies \sigma \models \QReachAbs{\ell'}{k}$.
\end{definition}

\begin{example}
\label{ex:lo-local-path-extension}
In the LO tree running example, single-step compatibility of $\QSReachXK{\cdot}{k}$ (\Cref{ex:reach-pred}) with $\extend{k}$ (\Cref{ex:traversal}) holds by construction---and this is usually the case---since this reachability predicate is formally defined
by the existence of a sequence of locations that follows $\extend{k}$ (see \Cref{ex:k-reach-def}).
At times, reachability is defined not through sequences of locations following $\extend{p}$, in which case the compatibility %
relies on a different argument.
This scenario arises when we analyze the traversal in~\crefrange{Ln:lo-locate-start}{Ln:lo-locate-end} and \crefrange{Ln:lo-contains-pred-start}{Ln:lo-contains-pred-end} (see~\Cref{sec:lo-traversal-full}).
\end{example}

\subsection{The Condition on Interference: Forepassed}
\label{sec:condition-on-inteference}
We now describe our main condition about how the writes in the execution, potentially interfering with traversals, affect the reachability predicate. In essence, the idea is that if a write $w$ reduces the reachability of a memory location $\ell$, then afterwards writes to $\ell$ are not allowed, unless they modify $\ell$ in a very specific way:
\begin{changebar}
by pointing only to locations that have already been reachable %
at some point in between these writes.
\end{changebar}

\begin{definition}[Forepassed]
\label{def:black-condition}
A write $w$ at time $t > \texbegin$ in the execution %
satisfies the \emph{forepassed} condition
if for every location $\ell$, either
\begin{enumerate}
	\item \label{cond:still} $\models_{t-1} \QReachAbs{\ell}{k} \implies \models_{t} \QReachAbs{\ell}{k}$ (that is, $\ell$'s reachability is not reduced by $w$);
	or
	\item \label{cond:pastreach} %
	for every %
	write $w'$ to $\ell$ in time $t' \in [t,\texend]$, if $w'$ writes a value $v$ to $\ell$ and $\extend{p}(\ell, v, \ell')$ for some $\ell'$,
	then $\Pastlimit{\QReachAbs{\ell'}{k}}{t-1}{t'}$ holds
 (that is, every subsequent write to $\ell$, including the current write $w$ (if it writes to $\ell$), points to a location that has been reachable %
 at some point from just before $w$ to just after $w'$.)
\end{enumerate}
\end{definition}

In the algorithms we consider, most interfering writes satisfy a stronger, and simpler, condition: that if $w$ reduces the reachability of $\ell$, then $\ell$ is not modified by any write,
including $w$: %
\begin{definition}[Strong Forepassed]
\label{def:preservation-modification}
A %
write $w$ at time $t$ in the execution satisfies \emph{strong forepassed} %
if for every location $\ell$, either %
\begin{enumerate}
\item $\models_{t-1} \QReachAbs{\ell}{k} \implies \models_{t} \QReachAbs{\ell}{k}$ (that is, $\ell$'s reachability is not reduced by $w$, as in \Cref{def:black-condition}); or
\setcounter{enumi}{2}
	\item \label{cond:immu} %
	no %
	write in time $[t,\texend]$  modifies $\ell$ (that is, $\ell$ is ``immutable'' from this time on; this includes the current write $w$, so $w$ cannot write to $\ell$ itself).
\end{enumerate}
\end{definition}
Note that condition (\ref{cond:immu}) is a special case of (\ref{cond:pastreach}) in the presence of single-step compatibility, but it is often conceptually simpler, so we often allude to it in our applications of the framework to different data structures.

When all the writes satisfy the strong forepassed condition, this corresponds to the preservation of reachability to locations of modification from~\citet{FeldmanE0RS18}.

\newcommand{\still}{\ref{cond:still}}
\newcommand{\immu}{\ref{cond:immu}}
\newcommand{\pastreach}{\ref{cond:pastreach}}

\begin{example}
\label{ex:lo-black-condition-holds}
In the LO tree running example, %
the (strong) forepassed condition holds w.r.t.\ $\QSReachXK{\cdot}{k}$.
The idea is that only deletions reduce the $k$-reachability of mutable locations, these locations become immutable afterwards, thanks to marking them, thereby satisfying the strong forepassed condition.
In more detail, we need only consider modifications to $\fnext$ fields, since $\fkey$ is immutable, and other fields are not involved in the definition of $\QSReachXK{x}{k}$, and thus cannot reduce the reachability of any location (hence satisfying condition~\still). These writes are:
\begin{inparaenum}
	\item \cref{Ln:lo-insert-write-new}, which modifies an unreachable (newly allocated) node, hence does not reduce the reachability of any location;
	\item \cref{Ln:lo-insert-write-previous}, which reduces the $k$-reachability of $\code{z}.\fkey$ for $k$ that is the inserted key.
However, the key field would not be modified later (it is immutable); %
	\item \cref{Ln:lo-delete-write-previous}, which reduces the $k$-reachability of \begin{changebar}$\code{s}$\end{changebar}. %
	However, \begin{changebar}\code{s}\end{changebar} is marked before the locks are released, and thus future
operations would refrain from modifying it (see the assertions in~\cref{Ln:lo-insert-write,Ln:lo-delete-write}). %
\end{inparaenum}
Note that this argument considers how writes affect reachability and the possibility of later writes, and need not resort to complex reasoning about how writes interleave with the traversal's reads (as typical linearizability proofs require). %
For an illustration of case \pastreach{} of the forepassed condition, see~\Cref{sec:cf-short}.
\yotamsmallnew{reading this, I don't really feel that the issue of writes vs the code arises here? maybe the remark above is enough?}
\end{example}
The forepassed condition is a property of the writes in the execution, so establishing it requires concurrent reasoning on interleavings of writes (but not on how reads interleave with writes). As illustrated by the example, the necessary reasoning is often very simple, as it does not require the correctness of traversals for the sake of proving the forepassed condition.
In other cases, establishing the forepassed condition can benefit from the properties of writes which are themselves proved based on the correctness of preceding traversals. This is in fact possible, using a proof by induction, as we explain in~\Cref{sec:circular}.
Overall this leads to simple proofs of the forepassed condition and, as a result, of traversal correctness.

\subsection{Main Theorem}
\label{sec:main-theorem}
We are now ready to state our main theorem that establishes traversal correctness from the ingredients above. 
We begin with a lemma that captures the important effect of the forepassed condition: that from whenever a location becomes reachable and onwards, the value it holds directs traversals only to locations that themselves have been reachable (although not necessarily at the same time).
\begin{lemma}
\label{lemma:red-condition}
Consider an execution in timespan $[\texbegin,\texend]$.
If $\QReachAbs{\cdot}{k}$ is single-step compatible with $\extend{k}{}$ (\Cref{sec:local-path-extension}) and all writes satisfy the forepassed condition (\Cref{def:black-condition}), %
then if a location is reachable, it will always afterwards point to a location that was reachable:
\begin{equation*}
	\forall t \in [\texbegin,\texend]. \ \forall t' \in [t,\texend]. \ \forall \ell,\ell'. \
		 \models_{t} \QReachAbs{\ell}{k} \land \extend{p}(\ell,\sigma_{t'}(\ell),\ell')
			\implies
		\Pastlimit{\QReachAbs{\ell'}{k}}{\texbegin}{t'}
\end{equation*}
\end{lemma}

\begin{proof}%
Let $t,t',\ell,\ell'$ be as in the premise of the lemma, and denote
$v = \sigma_{t'}(\ell)$.
\begin{changebar}
Our goal is to find $\tilde{t} \in [\texbegin,t']$ such that $\models_{\tilde{t}} \QReachAbs{\ell'}{k}$.
\end{changebar}

We consider two cases
depending on whether $\sigma_{t}(\ell) = \sigma_{t'}(\ell) = v$ holds. Let us first
assume that it does. Note that $\models_{t} \QReachAbs{\ell}{k}$ and
$\extend{p}(\ell,v,\ell')$ both hold. Thus, by single-step compatibility,
we get $\models_{t} \QReachAbs{\ell'}{k}$. Thus, letting $\tilde{t} = t$
concludes the lemma for this case. %

We now consider the case when $\sigma_{t}(\ell) \neq \sigma_{t'}(\ell) = v$.
There must exist a write $w$ in time $t_w \in (t,t']$ that modifies $\ell$ to
$v$ (that is, $\sigma_{t_w}(\ell) = v$ holds). Recall that all writes satisfy
the forepassed condition. %
Let us first assume that they all satisfy
\Cref{def:black-condition}.(\still{}) on $\ell$: i.e., no write in
$[\texbegin,\texend]$ reduces the reachability of $\ell$. Therefore, neither
do writes in $(t,t_w]$. In that case, knowing that $\models_{t}
\QReachAbs{\ell}{k}$ holds, we get that so does $\models_{t_w}
\QReachAbs{\ell}{k}$. %
From the premise we have that $\extend{p}(\ell,v,\ell')$ %
and since $\sigma_{t_w}(\ell) = v$
by
single-step compatibility
we get $\models_{t_w}
\QReachAbs{\ell'}{k}$. Overall, in this case, taking $\tilde{t} = t_w$ yields the desired.

Let us now consider the case when there is at least one write $w^\dagger$ at
time $t^\dagger \in (t,\texend]$ reducing the reachability of $\ell$. When
$t^\dagger \in (t_w, \texend]$, we establish the lemma analogously to the
previous case. Let $t^\dagger \in (t,t_w]$ hold. %
\Cref{def:black-condition}.(\pastreach) holds of
$w^\dagger$. Hence, for $w$, which occurs at $t_w > t^\dagger$ and writes a value $v$ satisfying $\extend{p}(\ell,v,\ell')$, we get that there exists 
\yotamsmallnew{note the change:}$\tilde{t} \in [t^\dagger-1,t'] \subseteq [\texbegin, t']$ such that $\models_{\tilde{t}} \QReachAbs{\ell'}{k}$, which
concludes the proof.
\end{proof}

\begin{theorem}
\label{thm:main-thm}
Consider an execution in timespan $[\texbegin,\texend]$ and a traversal
$\tau = (\ell_1,t_1),\ldots,(\ell_n,t_n),\ell_{n+1}$ defined through $\extend{p}$, such that %
$[t_1,t_n] \subseteq [\texbegin,\texend]$ and $\models_{\texbegin}
\QReachAbs{\ell_1}{k}$ hold. If $\QReachAbs{\cdot}{k}$ is single-step compatible
with $\extend{p}$, and all writes in the execution satisfy the forepassed
condition (\Cref{def:black-condition}), %
then $\tau$ is a correct traversal w.r.t.\ $\QReachAbs{\cdot}{k}$ and $\texbegin$. %
\end{theorem}
\begin{proof}%
We do a proof by induction on the length of the traversal.
By~\Cref{def:traversal-correctness}, we need to show
that $\models_{\texbegin} \QReachAbs{\ell_1}{}$ holds, and that
for every $i$, $1 \leq i \leq n$, it holds that
$\Pastlimit{\QReachAbs{\ell_{i+1}}{}}{\texbegin}{t_i}$. The former is the
base case of induction and holds trivially as a premise of the theorem. For the
latter, we let $t_0 = \texbegin$ for convenience of notation, and prove the
induction step: assuming the induction hypothesis
$\Pastlimit{\QReachAbs{\ell_{i+1}}{}}{\texbegin}{t_i}$, we show
$\Pastlimit{\QReachAbs{\ell_{i+2}}{}}{\texbegin}{t_{i+1}}$.

From the induction hypothesis we have that there there exists $\tilde{t}_{i} \in [\texbegin, t_i]$
such that $\models_{\tilde{t}_i} \QReachAbs{\ell_{i+1}}{}$. %
Hence,
since $t_{i+1} \geq t_i \geq \tilde{t}_i$,
we apply \Cref{lemma:red-condition} to obtain %
\begin{equation*}
 \models_{\tilde{t}_i} \QReachAbs{\ell_{i+1}}{k} \land \extend{p}(\ell_{i+1},\sigma_{t_{i+1}}(\ell),\ell_{i+2})
	\implies
\Pastlimit{\QReachAbs{\ell_{i+2}}{k}}{\texbegin}{t_{i+1}}.
\end{equation*}
According to the definition of traversals, %
we have $\extend{p}(\ell_{i+1},\sigma_{t_{i+1}}(\ell),\ell_{i+2})$, and so
the premise of the equation above holds. We conclude
$\Pastlimit{\QReachAbs{\ell_{i+2}}{k}}{\texbegin}{t_{i+1}}$, and with it the
induction step.
\end{proof}

\begin{example}
\label{ex:lo-succ-correctness}
We now apply our main theorem to deduce traversal correctness for traversals over successor links.
We showed in~\Cref{ex:lo-local-path-extension} that $\QSReachXK{\cdot}{k}$ is single-step compatible with $\extend{k}$, and in~\Cref{ex:lo-black-condition-holds} that interfering writes satisfy the forepassed condition \Cref{def:black-condition}. Therefore, any traversal with a base $\texbegin$ s.t.~$\models_{\texbegin} \QReachAbs{\ell_1}{k}$ starting from $\ell_1$ is a correct traversal.
In~\Cref{sec:lo-traversal-full} we use this fact to prove properties of the entire traversal in~\Cref{Fi:LO-code} (which starts from tree- and predecessor-links before it traverses successor links).

\ignore{
	From the local path extension (\Cref{sec:local-path-extension}) and forepassed interference (\Cref{def:black-condition}), our main theorem shows the correctness of a traversal along~\crefrange{Ln:lo-contains-succ-start}{Ln:lo-contains-succ-end} that starts from $\succroot$---if it reaches $x$ at time $t'$, then $\Pastlimit{\QSReachXK{x}{k}}{\texbegin}{t'}$,
	where $\texbegin$ is some time $\texbegin \leq t_1$.
	The traversal in the running example does \emph{not} start from the root, but
	the same theorem also shows the correctness of such a traversal if it starts not from $\succroot$ itself, but from a node $y$ s.t.\ $\models_{\texbegin} \QSReachXK{y}{k}$ for some base time $\texbegin$.
}
\ignore{
	Therefore, to prove the correctness of the entire traversal in the tree all that remains is to show that the traversal in~\crefrange{Ln:lo-contains-succ-start}{Ln:lo-contains-succ-end} starts from such a $y$ for some $t_0$ that is chosen accordingly, to which end we apply this theorem again, this time to the traversal along tree- and predecessor-links (\crefrange{Ln:lo-locate-start}{Ln:lo-locate-end} and \crefrange{Ln:lo-contains-pred-start}{Ln:lo-contains-pred-end}). This is explained in~\Cref{sec:lo-traversal-full}.
}
\end{example}

\subsection{Reachability with Another Field}
\label{sec:reachability-with-field}
The assertions in~\cref{Ln:lo-ContainsRetTrue,Ln:lo-ContainsRetFalseRem} are about some point in time in which a location was reachable \emph{and at the same time} a certain field ($\frem$) had a certain value (true or false). Our proof technique extends to properties involving reachability of an object \emph{and} the value of a single field in the following way:
Consider a location $\ell$ and a location $\fieldloc$ (intuitively, $\fieldloc$ is a field of an object that resides in $\ell$). We require the following condition, akin to the strong forepassed condition (\Cref{def:preservation-modification}), %
but that focuses only on $\ell,\fieldloc$:

\begin{definition}
\label{def:black-condition-field-extension}
A %
write $w$ at time $t$ in the execution %
satisfies the \emph{forepassed} condition
w.r.t.\ to locations $\ell,\fieldloc$ if
either
\begin{enumerate}
	\item $\models_{t-1} \QReachAbs{\ell}{k} \implies \models_{t} \QReachAbs{\ell}{k}$ (that is, $\ell$'s reachability is not reduced by $w$); or
	\item %
	no write in time $t' \in [t,\texend]$  modifies $\fieldloc$ (that is, $\fieldloc$ is ``immutable'' from this time on; this includes the current write $w$).
\end{enumerate}
\end{definition}

\begin{theorem}
\label{thm:field-extension}
Consider an execution in timespan $[\texbegin,\texend]$. If $t$ and $t'$ are
such that $\texbegin \leq t \leq t' \leq \texend$, $\models_{t}
\QReachAbs{\ell}{k}$ and $\sigma_{t'}(\fieldloc) = v$ hold, and all writes
satisfy \Cref{def:black-condition-field-extension} in the interval
$[\texbegin,\texend]$, then it is the case that $\Pastlimit{\QReachAbs{\ell}{k}
\land \fieldloc = v}{t}{t'}$ holds. %
\end{theorem}
\begin{proof}%
Let $t,t',\ell,\ell'$ be as in the premise. When $\sigma_{t}(\fieldloc) =
\sigma_{t'}(\fieldloc)$ holds, we have $\models_t (\QReachAbs{\ell}{k} \land
\fieldloc = v)$, and the theorem holds trivially. In the following, we consider
the case when $\sigma_{t}(\fieldloc) \neq
\sigma_{t'}(\fieldloc) = v$. There must exist a write $w$ in time $t_w \in
(t,\texend]$ that modifies $\fieldloc$ to $v$ (and then $\sigma_{t_w}(\fieldloc)
= v$ holds). Recall that all writes satisfy the forepassed condition. Let us
assume that there is a write $w'$ at $t'_w \in (t,t_w]$ reducing reachability of
$\ell$. By \Cref{def:black-condition-field-extension}.(2), no later write in
$[t'_w, t_w]$ modifies $\fieldloc$. Since $w$ modifies $\fieldloc$ at
time $t_w$, we arrive to a contradiction. Thus, all writes in $(t,t_w]$
satisfy \Cref{def:black-condition-field-extension}.(1): no write in
$(t,t_w]$ reduces the reachability of $\ell$. Knowing that
$\models_{t} \QReachAbs{\ell}{k}$ holds, we get that so does $\models_{t_w}
\QReachAbs{\ell}{k}$. Finally, when $\sigma_{t_w}(\fieldloc) = v$, we get
$\models_{t_w} (\QReachAbs{\ell'}{k} \land \fieldloc = v)$.
\end{proof}

\begin{example}\label{ex:field-extension}
In the LO tree, we use this extension to deduce properties such as $\Past{\QSReachXK{x}{k} \land \QXR{x}}$ in~\cref{Ln:lo-ContainsRetFalseRem} and similarly in~\cref{Ln:lo-ContainsRetTrue}. The condition holds for the reachability of $x$ the field and $x.\frem$ because when the reachability of $x$ is reduced, it is marked (see~\Cref{ex:lo-black-condition-holds}), so future writes refrain from modifying $x.\frem$ (\cref{Ln:lo-delete-write-rem}).
\end{example}  %

\section{Logical Ordering Traversal, The Full Story}
\label{sec:lo-traversal-full}
In this section, we give an overview of how our framework applies in proving the challenging past-reachability assertions, whose proof completes the linearizability proof of the LO tree (\Cref{sec:lo-linearizability}).
To this end, we consider an execution from when the operation begins, $t_\mathsf{begin}$,
 that at time $t$ reaches the following assertions appearing
in the proof outline for {\tt contains} in \Cref{Fi:LO-code}:
\begin{enumerate}[(i)]
	\item $\Pastlimit{\QSReachX{x}}{t_\mathsf{begin}}{t}$ at the traversal over tree- and predecessor-links (lines~\ref{Ln:lo-locate-start}--\ref{Ln:lo-locate-end},\ref{Ln:lo-contains-pred-start}--\ref{Ln:lo-contains-pred-end}); 
	\item $\Pastlimit{\QSReachXK{x}{k}}{t_\mathsf{begin}}{t}$ afterwards, at the traversal over successor links (lines~\ref{Ln:lo-contains-succ-start}--\ref{Ln:lo-contains-succ-end});
	\item $\Pastlimit{\QSReachXK{x}{k} \land \QXR{x}}{t_\mathsf{begin}}{t}$ and $\Pastlimit{\QSReachXK{x}{k} \land \lnot \QXR{x}}{t_\mathsf{begin}}{t}$ afterwards, after further reading the $\frem$ field (\cref{Ln:lo-ContainsRetFalseRem,Ln:lo-ContainsRetTrue}).
\end{enumerate}
We establish the first two assertions by applying~\Cref{thm:main-thm}, proving traversal correctness~(\Cref{def:traversal-correctness}) w.r.t.\ an appropriate reachability predicate and the base time $t_\mathsf{begin}$, with an $\extend{p}$ relation capturing the traversal. %
We establish the last assertions using~\Cref{thm:field-extension}, our extension for reachability with the value of a single field.
The proof of each assertion relies on the preceding ones; we now describe how these proofs progress using our framework.

\para{Case (i)}
We capture this traversal, %
using the following \textbf{extend} relation:
\begin{equation*}
\begin{array}{lr}
\begin{array}{l}
	\extend{}(\oloc{o}{\fkey},\cdot,\oloc{o}{\mathbf{f}})
	\\
	\extend{}(\oloc{o}{\mathbf{f}},\oloc{o'}{\fkey},\oloc{o'}{\fkey})
\end{array} 
& \qquad \qquad \mathbf{f} \in \{\fleft,\fright,\fprev\}.
\end{array}
\end{equation*}

For the \textbf{reachability predicate}, we take
$\QSReachX{x}$,
which holds of a state $\sigma$ iff there is a sequence of
locations $\ell_1,\ldots,\ell_{n+1}$ starting from $\ell_1 = (\succroot,
\fkey)$, ending in $\ell_{n+1} = (x, \fkey)$, which is connected via the successors
list: for every $i \in [0,n]$, if $\ell_i = (o,\fkey)$ then
$\ell_{i+1} = (o,\fnext)$, and if $\ell_i = (o,\fnext)$ with
$\sigma(\ell_i)=(o',\fkey)$ then $\ell_{i+1} = (o',\fkey)$.

\begin{figure}[H]
\includegraphics[scale=0.4]{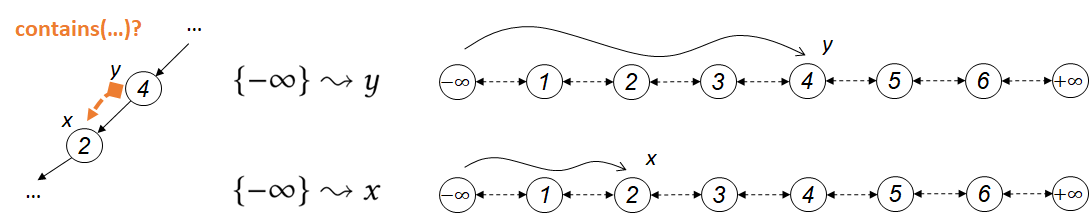}
\caption{Single-step compatibility of the tree traversal w.r.t.\ $\QSReachX{\cdot}$.}
\label{LO:tree-traversal-reach}
\end{figure}% The predicate $\QSReachX{\cdot}$ is
\textbf{single-step compatible} with the relation $\extend{}$. This is because it is an invariant that in $\QSReachX{}$-reachable nodes, $\fleft$-,
$\fright$- and $\fprev$-links point to other $\QSReachX{}$-reachable nodes: $\forall x,y. \ \QSReachX{x} \land x.\mathbf{f} = y \implies \QSReachX{y}$, for $\mathbf{f} \in \set{\fleft,\fright,\fprev}$. 
\Cref{LO:tree-traversal-reach} illustrates how a traversal moving across $\fleft$/$\fright$ pointers remain on nodes that are reachable in the successors list.
This invariant holds
since \code{insert} first links a node to the successors list, before attaching it to the tree or the successors list; and
\code{remove} unlinks from the successors list only after it unlinks from the
tree\footnote{This relies, in~\cref{Ln:lo-remove2-done}, on the invariant that
the sole parent of $x$ is $x.\fparent$, when $x$ is unlocked.
\begin{changebar}
Note that modifications to the $\fparent$ field are protected by the parent's \code{treeLock} (in other words, a linked node's parent field is written to only when the parent's treeLock is held---which occurs during removal and rotations), and thus $\code{c}$, $\code{n.left}$, and $\code{n.right}$ do not have to be locked in~\code{removeFromTree} since only their $\fparent$ field is modified.
\end{changebar}} %
and the predecessors list.
Other operations, such as rotations, may unlink a node
from the tree but not from the successors list. %
To be able to apply the framework, we need to prove that interfering writes satisfy the \textbf{forepassed} condition. Indeed, the writes in LO satisfy case \still{} or case \immu{} of the forepassed condition, with a gist similar to~\Cref{ex:lo-black-condition-holds}: the only write that reduces the reachability of a location is the removal of a node in~\cref{Ln:lo-delete-write-previous}, reducing the reachability of \code{y}, but this node is marked removed and no further writes to it will occur, satisfying case \immu{}.
From these ingredients,
\Cref{thm:main-thm} \textbf{yields the desired assertion}.
Formally, we consider any traversal $\tau =
(\ell_1,t_1),\ldots,(\ell_n,t_n),\ell_{n+1}$ w.r.t.\ $\extend{}$ occurring within an execution in
timespan $[t_\mathsf{begin}, t]$, so that $\ell_1 = (\succroot,
\fkey)$, ending in $\ell_{n+1} = (x, \fkey)$.
By~\Cref{thm:main-thm},  $\tau$ is correct w.r.t.
$\QSReachX{\cdot}$ and $t_\mathsf{begin}$, which concludes the case (i).

\para{Between (i) and (ii)}
The choice of the reachability predicate in assertion (i) is with the aim of proving assertion (ii), which concerns $\QSReachXK{\cdot}{k}$ in the successors list. In assertion (i) we used a different reachability predicate, of plain-reachability through successor links, without considering any specific key.
The reason is that the traversal over tree- and predecessor-links visits nodes in an order that is, in some sense, ``random access'' into the successors list, and does not respect the search for $k$ in the successors list.
However, assertion (i) is important for proving assertion (ii). %
The necessary ``glue'' is the following observation of what holds in between them, before the first iteration in~\cref{Ln:lo-contains-succ-start}:
Let $z$ be the value of \code{x} just after the loop at~\crefrange{Ln:lo-contains-pred-start}{Ln:lo-contains-pred-end}.
From this loop's condition, necessarily $\QXK[\leq]{z}{k}$.
From assertion (i), there is some timestamp $t' \geq t_\mathsf{begin}$ when $\QSReachX{z}$ holds.
The successors list is sorted and contains
unique values, an invariant that can be established from the assertions as (see~\Cref{sec:lo-linearizability}).
 Hence, it also holds that $\QSReachXK{z}{k}$ at time $t'$.
Thus assertion (ii) holds just before the loop in~\crefrange{Ln:lo-contains-succ-start}{Ln:lo-contains-succ-end}.
Our goal now is to prove this assertion also when this loop executes and \code{contains} traverses successor links.

\para{Case (ii)}
To prove this assertion we consider a traversal over (only) successor links that starts from $z$---where the traversal over tree- and predecessor-links left off---strictly after that traversal:
$
\tau' = (\ell'_1, t'_1), (\ell'_{2}, t'_{2}), \ldots, (\ell'_m, t'_m), \ell'_{m+1}$, 
where $\ell'_1 = (z, \fkey)$ %
and $t'_1 \geq t'$ (as $t'$ occurred sometime during the previous traversal).
As a traversal over successor links, it
visits locations connected by the {\bf $\extend{k}$} relation from
\Cref{ex:traversal}. We study its correctness w.r.t. the {\bf reachability
predicate} $\QSReachXK{\cdot}{k}$ from \Cref{ex:k-reach-def}. As we have shown
in \Cref{ex:lo-local-path-extension},  $\QSReachXK{\cdot}{k}$ is {\bf single-step
compatible} with $\extend{k}$. We also proved in
\Cref{ex:lo-black-condition-holds} that the LO tree's writes satisfy the {\bf forepassed} condition.
At time $t'$, the first location is $k$-reachable: $\QSReachXK{(z,\fkey)}{k}$.
From these, by
\Cref{thm:main-thm} we get that $\tau'$ is a correct traversal w.r.t.
$\QSReachXK{\cdot}{k}$ and the base time $t'$. By~\Cref{def:traversal-correctness},
$\Pastlimit{\QSReachXK{x}{k}}{t'}{t}$ holds, and in particular $\Pastlimit{\QSReachXK{x}{k}}{t_\mathsf{begin}}{t}$,  which concludes the case (ii).

\para{Case (iii)}
We apply the extended framework of \Cref{sec:reachability-with-field} in our
proofs of $\Past{\QSReachXK{x}{k} \land \QXR{x}}$ and $\Past{\QSReachXK{x}{k}
\land \lnot \QXR{x}}$.
From assertion (ii), there is a point in time $\tilde{t} \in [t_\mathsf{begin},t]$ where $\QSReachXK{x}{k}$ holds.
We consider any possible execution with timespan
$[\tilde{t}, t]$. %
Let
$v$ be the value of $x.{\tt rem}$ returned by the read at
\cref{Ln:lo-ContainsReadRem}, i.e. $v = \sigma_{t'}(x.{\tt rem})$.
As we have shown in \Cref{ex:field-extension}, the LO tree's writes
satisfy the {\bf forepassed} condition w.r.t. $\ell$ and $x.{\tt rem}$. By
\Cref{thm:field-extension}, $\Past{\QSReachXK{x}{k} \land
\QXR{x} = v}$ holds, which concludes the case (iii).

\section{Discussion: On Proving the Forepassed Condition}
\label{sec:circular}
The forepassed condition (\Cref{def:black-condition}) is the key requirement to algorithm implementations in our framework. While it is simple to establish in the LO tree (\Cref{sec:lo-traversal-full}),
in general this requires reasoning about concurrent executions.
However, this task can be simplified by relying on assertions showing the correctness of writes.
Since proofs of traversal correctness with our framework are carried out by induction on the length of concurrent executions, inductive arguments for the forepassed condition can be integrated into the proof. This introduces a curious circularity: the forepassed condition on the prefix of an execution is used to conclude correctness of the corresponding traversal, which in turn can be leveraged in justifying the forepassed condition on a longer prefix of the execution. The integration is possible because justifying traversal correctness after a prefix of an execution requires the forepassed condition to hold only on that prefix. 
A similar approach has been previously proposed by~\citet{FeldmanE0RS18}.

\ignore{
	The main condition of our framework requires a temporal invariant over writes. At times, establishing this condition is significantly simplified by relying on the assertions that express the correctness of individual writes. If these assertions are justified while relying on our framework to prove traversal correctness, the argument is seemingly circular---traversal correctness relies on the assertions, and the assertions rely on traversal correctness. In fact, it is possible to unleash this kind of powerful arguments by an induction over time, and this is a very useful feature of our framework (as in~\citet{DBLP:conf/wdag/FeldmanE0RS18}\sharon{citation broken}).
	The reason is that the correctness of a write uses only  guarantees provided by traversals that occurred \emph{earlier}, and the induction over reads in a traversal (in the proof of~\Cref{thm:main-thm}) uses properties of interfering writes that occurred \emph{before} the read. Thus the entire proof---of traversal correctness and the writes---becomes a mutual induction over time.
	(The LO tree does not benefit from this sort of argument because the proof of the modifiers does not rely on traversal correctness. A ``circular'' argument is however useful in the proofs of the Contention-Friendly Tree, in~\Cref{sec:cf-short}, and the Citrus Tree, in~\Cref{sec:citrus-short}.)
}

\section{Additional Case Studies}
\label{sec:additional}

\subsection{List-Based Structures}
Our method can prove all the list-based structures handled by the local view framework~\cite{FeldmanE0RS18}: Lazy List~\cite{HellerHLMMS05}, lock-free list~\cite[Chapter 9.8]{TAOMPP}, and and lock-free skiplist~\cite[Chapter 14.4]{TAOMPP}. This is because the traversals in all these examples are single-step compatible with the reachability predicate (similar to~\Cref{ex:lo-local-path-extension}), and the preservation condition of~\citet{FeldmanE0RS18} is a special case of forepassed interference (see~\Cref{def:preservation-modification}).

\subsection{Contention-Friendly Tree with Backtracking}
\label{sec:cf-short}
The contention-friendly tree~\cite{EuroPar13:CFTree,PFL16:CFTree} is a self-balancing binary search tree, in which traversals operate without synchronization, and rotations are performed by allocating a copy of the rotated node (see~\Cref{Fi:CF-Rotation}). The linearizability proof of this tree uses the $k$-reachability predicate $\QReachXK{x}{k}$, meaning that a node $x$ is on a path from the $\rootobj$ in a standard tree binary search for key $k$~\cite{FeldmanE0RS18}.
\begin{figure}
\begin{subfigure}{.65\textwidth}
\includegraphics[scale=0.44]{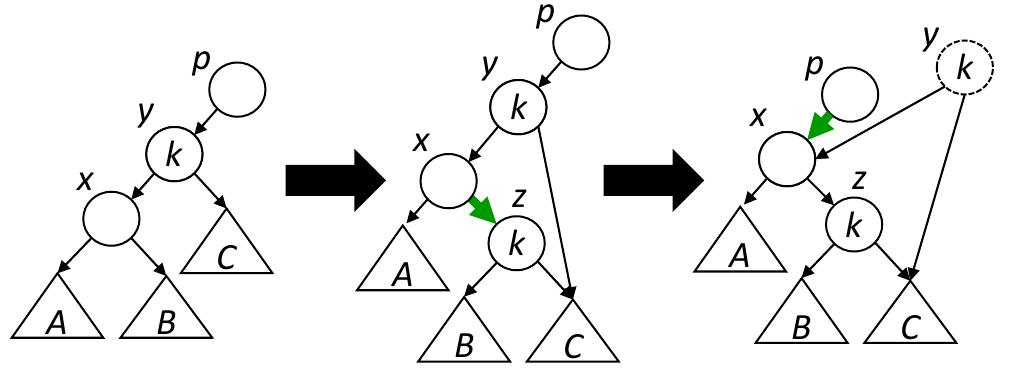}
\caption{\footnotesize Right rotation of $y$ in the Contention-Friendly Tree, from~\cite{FeldmanE0RS18}.  (The bold green link is the one written in each step.
The node with a dashed border has its $\frem$ bit set.)
}
\label{Fi:CF-Rotation}
\end{subfigure}
\ \ \ 
\begin{subfigure}{.3\textwidth}
\centering
\includegraphics[scale=0.44]{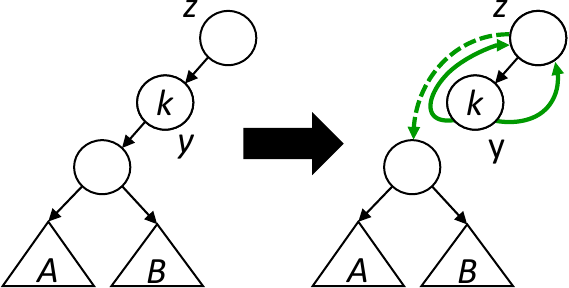}
\subcaption{\footnotesize Bypassing $y$ and backtracking from it by pointing from $y$ to its parent $z$.}
\label{Fi:Backtracking}
\end{subfigure}
\caption{Operations in the contention-friendly tree.}
\end{figure} \citet{FeldmanE0RS18} proved traversal correctness of a variant of this algorithm, but their proof cannot handle \emph{backtracking}. In the original version~\cite{EuroPar13:CFTree,PFL16:CFTree}, when a node $x$ is physically removed, its $\fleft$/$\fright$ pointers are modified to point to its parent (see~\Cref{Fi:Backtracking}). In this way, a traversal reaching a physically removed node backtracks until it can continue from a node
still linked to the tree~\cite{PFL16:CFTree}, without requiring traversals to perform explicit
synchronization or validation steps~\cite{Bronson:2010}.
This backtracking-like operation is inherently problematic for the framework of~\citet{FeldmanE0RS18} because it breaks their temporal acyclicity requirement: what was once a child of a node is now its parent.

We apply our framework to prove traversal correctness even in the presence of backtracking. %
For a key $k$, we capture the traversal searching for $k$ using the following $\extend{k}$ relation, defined 
to be true iff it is one of the following cases:
\begin{align*}
	&\extend{k}(\oloc{o}{\fkey},m,\oloc{o}{\fright}) &\quad \mbox{if $m<k$}
	\\
	&\extend{k}(\oloc{o}{\fkey},m,\oloc{o}{\fleft}) &\quad \mbox{if $m>k$}
	\\
	&\extend{k}(\oloc{o}{\fleft},\oloc{o'}{\fkey},\oloc{o'}{\fkey}) &
\\
	&\extend{k}(\oloc{o}{\fright},\oloc{o'}{\fkey},\oloc{o'}{\fkey}) &
\end{align*}
We define the $k$-reachability predicate through sequences of locations that follow extend:
$\QReachXK{x}{k}$ holds in state $\sigma$ if there is a sequence of locations $\ell_0,\ldots,\ell_n$ s.t.\ $\ell_0 = \oloc{\rootobj}{\fkey}$, $\ell_n = x$, and $\forall i < n. \ \extend{k}(\ell_i,\sigma(\ell_i),\ell_{i+1})$.
These definitions exactly follow 
a binary search in the tree, that is: if $k$ is greater (smaller) than the current key, the path continues through the right (left) child respectively. Note that the path does not continue after finding the target key.
Since $\QSReachXK{\cdot}{k}$ is defined using $\extend{k}$, their {\bf single-step compatibility} is evident.

The {\bf forepassed} condition holds because in this algorithm,
when the $k$-reachability of a node is reduced, the node is marked, so future operations do not modify it (similar to~\Cref{ex:lo-black-condition-holds}, \emph{except} the backtracking modifications---which do modify the node's pointers \emph{after} it is no longer reachable.
However, these modifications satisfy case \pastreach{} of the forepassed condition, because they point to the parent, which has been $k$-reachable itself at the time that its child's reachability was reduced.
It is important to note that rotations, which could reduce the $k$-reachability of the node rotated downwards (because binary searches now encounter the node rotated upwards first and can continue in the other direction), satisfy the forepassed condition in the CF tree, because rotations in the CF tree use a \emph{newly allocated} node to represent the node rotated down (see~\Cref{Fi:CF-Rotation}).
See~\refappendix{sec:cf-backtrack} for the code and a detailed discussion of the traversal correctness proof of this algorithm.
\subsection{Citrus Tree}
\label{sec:citrus-short}
The Citrus tree~\cite{Arbel:2014} is a concurrent binary tree implementing a key-value map with
the standard operations \code{insert(}$k$,$d$\code{)}, \code{delete(}$k$\code{)}, and \code{contains}($k$\code{)} operations. %
The tree's nodes include a
\code{rem} boolean field indicating logical removal (like in the LO tree),
and a \code{tag} integer field, used to prevent an ABA problem
due to multiple nullifications of the \code{left} field upon insertion. (Nodes also contain \code{key}, \code{data}, \code{left}, and \code{right} fields.) %
The operations of this concurrent map performs lock-free binary search in the tree. %
In the following, we discuss the use of our framework in proving their correctness (more details are given in~\refappendix{app:citrus}).

\begin{figure}
\begin{center}
\centering
\includegraphics[scale=0.35]{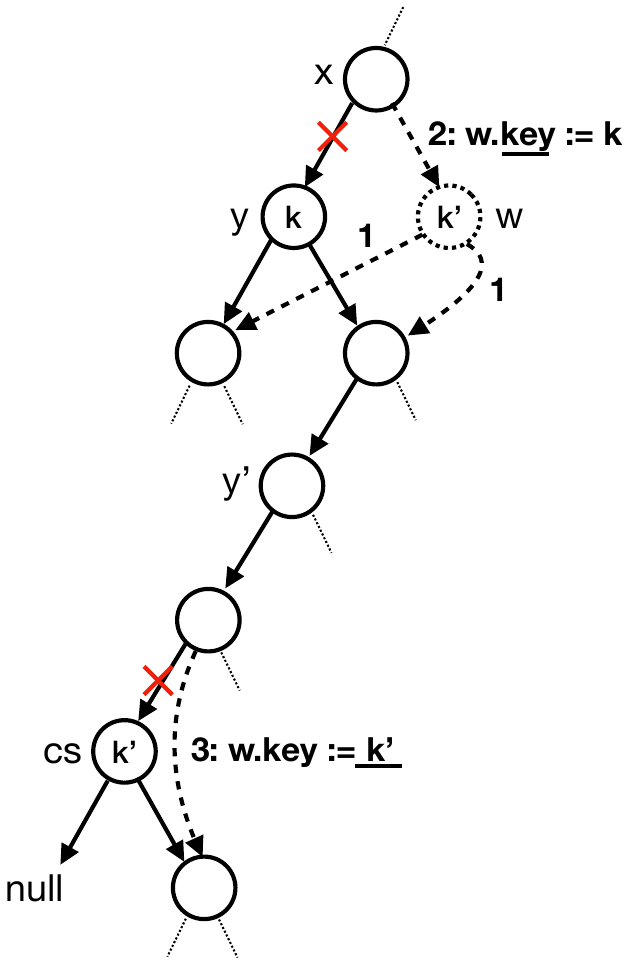}
\caption{\footnotesize Removing a node $y$ with two children by creating a copy $w$ of the node $cs$ containing the smallest key $k'$ bigger than the key $k$ of $y$ (the \emph{successor}). Dashed edges represent changes of pointers made during this removal. They are labeled with integers that show their order and updates of the ghost field \ghost{\code{key}}. %
}
\label{fig:citrus:delete}
\end{center}
\end{figure}
 
\begin{changebar}
The most intricate part of this algorithm is the physical removal of a node with two children as part of \code{delete}.
(Indeed, many concurrent tree algorithms with optimistic traversals refrain from physically removing nodes until they have only one child~\cite[e.g.][]{PFL16:CFTree,Bronson:2010}.)
In the Citrus tree, this is done as pictured in~\Cref{fig:citrus:delete}. (The assignments to the field \ghost{\code{key}} should be ignored for now.)
Let $y$ be a node with two children. The operation finds $cs$, the successor of $y$ in the tree, and creates a copy $w$ of $cs$. It then performs the following mutations (see the labels in~\Cref{fig:citrus:delete}):
\begin{inparaenum}
	\item setting $w$'s left and right children to be the same as in $y$;
	\item linking $w$ to the tree as the left child of $y$ (at this point $y$ is unlinked from the tree); and
	\item finally, unlinking $cs$ from the tree.
\end{inparaenum}
Nodes are first marked logically deleted before being unlinked from the tree. Every write to a node is protected by a lock associated to that node. Also, crucially, unlinking $cs$ from the tree is guarded by an RCU lock~\cite{rcu,rcu-thesis,userrcu} that synchronizes this write with traversals executing in parallel: it blocks this write until all the traversals that have already started finish.
\end{changebar}

\begin{changebar}
Proving that removals of nodes with two children do not hinder the results of concurrent lock-free traversals is quite tricky.
We show that the lock-free traversals in Citrus are \emph{not} correct with respect to the standard $k$-reachability predicate $\QReachXK{\cdot}{k}$ (\Cref{sec:citrus-short-standard-reach-fails}), leading us to defined a weaker version of the reachability predicate (\Cref{sec:citrus-short-weak-reachability-pred}), to which the framework applies (\Cref{sec:citrus-short-applying}). Correctness with respect to this weaker $k$-reachability predicate is however enough to prove linearizability (\Cref{sec:citrus-short-linearizability}).
\end{changebar}

\subsubsection{Warmup Attempt: Traversals Fail Standard Reachability}
\label{sec:citrus-short-standard-reach-fails}

\begin{changebar}
We first show that our framework (\Cref{thm:main-thm}) fails to prove traversal correctness w.r.t. $\QSReachXK{\cdot}{k}$ (which follows a binary search in the tree like in the previous case studies), and that this is due to the fact that the traversals are not correct w.r.t. this reachability predicate.

The traversal is single-step compatible with $\QSReachXK{\cdot}{k}$. %
However, the forepassed condition w.r.t. $\QSReachXK{\cdot}{k}$ does not hold.
When a node with two children is removed as in~\Cref{fig:citrus:delete}, the write that links $x$ to $w$ reduces reachability in the subtree of $w$. Specifically, it no longer holds that $\QSReachXK{\alpha.\fleft}{\tilde{k}}$ for every node $\alpha$ in the right subtree of $w$ and for every $\tilde{k}$ in the interval $(k,k']$ (upon this modification, $\tilde{k}$-paths go to the left of $w$ or stop at $w$, rather than going to the right as before).
However, these fields \emph{can} be modified afterwards, in a way that increases reachability, thereby \emph{violating the forepassed condition}.
This occurs when a new node $\beta$ is inserted as the left child of some $\alpha$, thus modifying $\alpha.\fleft$ to point to $\beta$ which is \emph{not} $\tilde{k}$-reachable at the point of insertion and of course also not before.
(Such a node $\beta$ can be inserted when the left child of $\alpha$ is deleted and it has two children, in which case $\beta$ serves as a copy of the successor of the left child of $\alpha$. Note that this successor is not $cs$, so this additional removal operation is possible concurrently.)
This constitutes a violation the forepassed condition, and~\Cref{thm:main-thm} does not apply.

In fact, it is not only that our main theorem cannot prove traversal correctness, but traversal correctness \emph{does not hold} with respect to the standard $k$-reachability predicate.
In the same scenario, a $\tilde{k}$-traversal that happens to reside at $y$ when the modification linking $w$ occurs and continues from there may reach the new node $\beta$ from above although $\beta$ was never $\tilde{k}$-reachable.
\end{changebar}
We use our framework to show that the traversals in this algorithm are correct w.r.t. $\QReachXKgap{\cdot}{k}$, a \emph{weaker} reachability predicate.
\begin{changebar}
Intuitively, $\QReachXKgap{x}{k}$ %
allows %
some searches to ``take a wrong turn''---in a way we make accurate hereafter---due to a concurrent removal of a node with two children.
The assertion $\Past{\QReachXKgap{x}{k}}$ implied by traversal correctness makes it possible to infer $\PReachXK{x}{k}$ (the form of reachability that underlies the abstraction function in the linearizability proof) in special cases, such as when $x$ is the endpoint of the traversal.
We first explain how $\QReachXKgap{\cdot}{k}$ is defined, then why our framework applies to this reachability predicate, and then sketch how $\QReachXK{\cdot}{k}$ can be derived from it as necessary.
\end{changebar}

\subsubsection{Weak Reachability Predicate Using Ghost State}
\label{sec:citrus-short-weak-reachability-pred}
The predicate $\QReachXKgap{\cdot}{k}$ is defined on top of an instrumentation of the algorithm with \emph{ghost code} that statically captures the ways traversals looking for a key $k$ can deviate from standard $k$-search paths (defined by the predicate $\QReachXK{\cdot}{k}$). Looking at~\Cref{fig:citrus:delete}, we refer to the deviation that occurs after linking $x$ to $w$ and define this predicate such that $\QReachXKgap{\ell}{\tilde{k}}$ holds for keys $\tilde{k}\in (k,k']$ and locations $\ell$ in the right sub-tree of $w$ even after $x$ is linked to $w$ (the problematic case mentioned above). Technically, we introduce a \emph{ghost} field \ghost{\code{key}} for every node in the tree, which does not change the actual behavior of the algorithm. The field \ghost{\code{key}} equals \code{key} unless the node is a copy $w$ introduced during the physical removal of a node $y$ with two children (cf.~\Cref{fig:citrus:delete}): \ghost{\code{key}} is set to $k$ atomically with linking $x$ to $w$, and it is set to $k'$ atomically with unlinking $cs$ from the tree (thus becoming equal to \code{key}).
\begin{changebar}
We refer to the last operation as ``collapsing'' the ghost interval; this indicates the end of this operation---from this point on, traversals are \emph{not} permitted to steer off  course because of this operation. (Collapsing the ghost interval is necessary to be able to infer interesting properties of $\QReachXK{\cdot}{k}$ out of weak reachability.)
This ghost state is defined \emph{per traversal}: a traversal starts with a copy of the state in which there is no ghost state ($\ghost{\code{key}} = \code{key}$ in all nodes), and a write modifies the ghost state of all the traversals that have already started and did not terminate.
\end{changebar}

The predicate $\QReachXKgap{x}{k}$ holds in a state $\sigma$ if and only if there is a sequence of
locations $\ell_1,\ldots,\ell_{n+1}$ starting from $\ell_1 = (\rootobj,\fkey)$, ending in $\ell_{n+1} = (x, \fkey)$, and connected via the \code{left} and \code{right} fields in the following manner:
for every $i=1,\ldots,n$,
\begin{align}
\label{eq:citrus-reachability-def}
\begin{split}
&\mbox{if $\ell_i = (o,\code{left})$ or $\ell_i = (o,\code{right})$ with $\sigma(\ell_i)=(o',\fkey)$, then} \\
&\hspace{2cm}\ell_{i+1} = (o',\fkey), \mbox{ and} \\
&\mbox{if $\ell_i = (o,\code{key})$, $\sigma(o,\code{key})=m$, $\sigma((o,\ghost{\code{key}}))=\ghost{m}$, then } \\
&\hspace{2cm}\ell_{i+1}=(o,\code{right})\mbox{ if $(k>\ghost{m} \land k\neq m) \lor (k=m \land \ghost{m}\neq m),$ and} %
 \\
&\hspace{2cm}\ell_{i+1}=(o,\code{left})\mbox{ if $k<m$}.
\end{split}
\end{align}

Note that standard $k$-reachability $\QReachXK{\ell}{k}$ implies that $\QReachXKgap{\ell}{k}$.
However, this predicate is weaker than $\QReachXK{\ell}{k}$ because:
\begin{inparaenum}
	\item the search for $k \in (\ghost{m},m)$ can go either left or right, and
	\item the search for $m$ can continue (to the right) after finding the key ($k=m$) in case the ghost key and the real key are different.
\end{inparaenum}
The updates on the ghost field \ghost{\code{key}} are visible only to the traversals executing in parallel (this is related to the use of the RCU lock). In more detail, each operation has its own instance of \ghost{\code{key}} for every node, and an update to this field should be read as modifying all the instances of the traversals executing in parallel atomically in one shot.

\subsubsection{Applying the Framework}
\label{sec:citrus-short-applying}
\begin{changebar}
First, the predicate $\QReachXKgap{\cdot}{k}$ is {\bf single-step compatible} with the standard $\extend{k}$ relation corresponding to binary search tree traversals (formally defined in~\refappendix{sec:cf-backtrack}). Essentially, extending a sequence of locations satisfying the relationship of~\Cref{eq:citrus-reachability-def} with a another location chosen according to binary search continues to satisfy~\Cref{eq:citrus-reachability-def}.
\end{changebar}

Second, showing that all the writes in this algorithm satisfy the strong {\bf forepassed} condition w.r.t. $\QReachXKgap{\cdot}{k}$ is relatively straightforward: %
Writes made within an \code{insert} operation %
only increase the reachability of locations. %
\begin{changebar}
Almost all the writes in \code{delete} satisfy forepassed for the usual reason (cf.~\Cref{ex:lo-black-condition-holds}): they reduce the reachability only of nodes that are marked logically deleted, which would never be modified. %
The only tricky write is the last step in the removal of a node with two children, unlinking the node $cs$ in~\Cref{fig:citrus:delete}, because it ``collapses'' the ghost field \ghost{\code{key}} $k'$ to the actual key $k$.
This reduces $\QReachXKgap{\cdot}{\tilde{k}}$ for all the nodes in the right sub-tree of $w$ and $\tilde{k} \in (k,k']$, which, when $k \neq k'$, is problematic because they can later be modified (as explained above in why $\QReachXK{\cdot}{k}$ does not satisfy forepassed).
This is where the RCU synchronization comes into play, together with our definition of the ghost code: if $\ghost{\code{key}} \neq \code{key}$, then concurrently with the traversal there was a write introducing this ghost. But then this traversal must terminate \emph{before} the collapsing of the ghost interval, because the RCU synchronization waits for existing traversals to terminate before performing this write. Thus, this \emph{reduction} of reachability---from $\ghost{\code{key}} \neq \code{key}$ to $\ghost{\code{key}} = \code{key}$---cannot occur.
Note that this argument uses the fact that the ghost state is per traversal and introduced by a write only in the existing, concurrent, traversals.

We can now apply~\Cref{thm:main-thm} to deduce that traversals in the Citrus tree are correct w.r.t.\ $\QReachXKgap{\cdot}{k}$.
\end{changebar}

\subsubsection{Proving Linearizability}
\label{sec:citrus-short-linearizability}
\begin{changebar}
So far we have established that every location $x$ the traversal reached satisfies $\Past{\QReachXKgap{x}{k}}$. %
To prove linearizability, we use this fact in order to infer properties of the standard reachability predicate. %
For example, in~\code{contains}, $\Past{\QReachXK{x}{k}}$ holds when $x$ is the \emph{endpoint} of the traversal (the last node it reached).
This can be inferred from $\QReachXKgap{x}{k}$ as follows.
\end{changebar}
This implication is immediate for traversals that do not pass through copies of nodes $w$ like in~\Cref{fig:citrus:delete} because the ghost field \ghost{\code{key}} can differ from \code{key} only for such nodes. Moreover, a traversal passing through such a node can deviate from a standard search path only when looking for a key $\tilde{k}\in (k,k']$. If it ends in $\NULL$, then $\tilde{k}\in (k,k')$ and $\QReachXK{\NULL}{\tilde{k}}$ was indeed true in the past (the algorithm ensures that $k'$ is the smallest key bigger than $k$).
Otherwise, if it ends in a non-null node $x$, then $\tilde{k}=k'$ and $x$ is the node $cs$ in~\Cref{fig:citrus:delete}. Again, %
this node was $k$-reachable in the past, concluding that indeed $\Past{\QReachXK{x}{k}}$, as desired.
\begin{changebar}
Slightly different properties are required for the linearizability of~\code{insert} and~\code{delete}. For instance, inserting to a left child uses \code{tag} validation to infer the that the modified node is reachable \emph{now} w.r.t. the \emph{standard} reachability predicate 
(from $\Past{\QReachXKgap{x}{k}}$). %
These are discussed in~\refappendix{app:citrus}.
The correctness of the traversal that searches for the successor node (as part of the operation that removes a node with two children), which is subject to different reachability patterns from the one in~\code{contains} (and serves a different purpose), is also discussed in~\refappendix{app:citrus}.
\end{changebar}   
\section{Related Work}
\label{sec:related}
\begin{changebar}
A couple of prior methods~\cite{PODC10:Hindsight,FeldmanE0RS18} also prove traversal correctness by reasoning strictly about how the algorithm's writes modify the memory state (i.e., by considering executions as sequences of interleaving writes, rather than interleavings of writes with traversals' reads) plus static \emph{sequential} properties of the traversal code.
Our method is more general:
the hindsight lemma~\cite{PODC10:Hindsight} is specific to linked lists; in particular, they do not clearly divide between the invariants used for traversal correctness and the rest of the proof.
\citet{FeldmanE0RS18} can handle algorithms beyond the list, but require that pointers always form an acyclic structure. %
In this work we relieve traversal correctness from the acyclicity condition, and instead build on single-step compatibility (\Cref{sec:local-path-extension}).
This allows us to prove examples that~\citet{FeldmanE0RS18} cannot handle, including the contention-friendly tree with backtracking and the Logical-Ordering tree (which includes in-place rotations).
Our framework also has the benefit of significantly simpler theory behind it.
The proof of our main theorem (\Cref{thm:main-thm}) is more similar to the hindsight lemma, by induction over the traversals' reads and case-splitting on whether the link was modified, rather than induction over interleaving writes as in~\citet{FeldmanE0RS18}.
\end{changebar}

Other works aim to simplify elements of the linearizability other than traversal correctness. %
\citet{DBLP:conf/wdag/Lev-AriCK15} harness properties of sequential executions in the linearizability proof. Their methodology relies on \emph{base points}, points during the concurrent execution where certain predicates hold, thus necessitating concurrent reasoning. When applying their framework to the lazy list, they rely on the tricky concurrent reasoning from previous works~\cite{PODC10:Hindsight,LazyListSafety} to establish base points. Our work is thus complementary, aiming to simplify concurrent reasoning of the sort that could also be used to establish base points.

\begin{changebar}
The Edgeset framework~\cite{DBLP:journals/tods/ShashaG88} proves the linearizability of concurrent search algorithms based on the notion of \emph{keyset}, which is reminiscent of $k$-reachability. This framework has recently been the algorithmic basis for the mechanizations by~\citet{DBLP:journals/pacmpl/KrishnaSW18,DBLP:conf/pldi/KrishnaPSW}, where the main technical contribution is showing how to obtain a mechanized proof of the Edgeset arguments in the Iris separation logic~\cite{DBLP:journals/jfp/JungKJBBD18} using the novel notion of flows (cf.~\cite[][\S7]{DBLP:conf/pldi/KrishnaPSW}). This development is orthogonal to our work, where we focus on the algorithmic essence of correctness, independent of a particular assertion language or program logic.
\end{changebar}
\citet{DBLP:journals/tods/ShashaG88} provide three algorithmic templates, and conditions for when these templates guarantee linearizability.
Of the three templates, the one closest to optimistic traversals is the link template.
\begin{changebar}
The Edgeset condition for the link template requires that a node gets \emph{accessed} for key $k$ \emph{only when} it is $k$-reachable or has an outgoing path leading to a $k$-reachable node.
This condition does not hold in our examples, where a node reached by $\code{contains}(k)$ may no longer be reachable, nor lead to a $k$-reachable node afterwards.\footnote{
	For example, consider the lazy list~\cite{HellerHLMMS05} (which is essentially the list traversal in the LO tree---also see~\Cref{sec:additional}). Suppose the list is $A \mapsto B \mapsto C \mapsto D$ and a read-only $\code{contains}(D)$ has read the pointer $A \mapsto B$. If now $B$, $C$, and $D$ are removed---in this order---by another thread, the $\code{contains}(D)$ will traverse the path $B \mapsto C \mapsto D$. Such a traversal is prohibited by the link template, because the traversed nodes are not reachable nor lead to a reachable node, %
	and so the $\code{contains}(D)$ is not allowed to access $B$, and similarly for $C$
	(cf.~\citet[][\S4.5]{DBLP:journals/tods/ShashaG88}).
	Still, $C$ and $D$ \emph{have been} reachable, and so such accesses are allowed by our framework.  %
} In contrast, the forepassed condition allows such accesses, but only constrains later modifications to such locations. %
The fundamental difference between the forepassed condition and the Edgeset conditions is that their conditions guarantee the existence of a path \emph{now}, whereas the forepassed condition applies to algorithms where it can only be shown that a path existed \emph{at some point}.
Furthermore, proving the link condition involves reasoning about interleavings of both reads and writes, whereas our forepassed condition involves interleavings of writes only.
Indeed, \citet{DBLP:conf/pldi/KrishnaPSW} use a strengthening of the Edgeset condition for the link technique, requiring that the set of nodes that is $k$-reachable or leads to a $k$-reachable node is never reduced by interleaving writes. This simplifies the reasoning, but is also too strong for optimistic traversals.
\end{changebar}

The designers of certain data structures proved correctness in ways that share some of the structure of the argument in our framework. \citet{Arbel:2014} prove that every node accessed in a traversal has been reachable at some point (Lemma 1), but this is shown for plain reachability rather than $k$-reachability, which is essential in binary search trees (e.g.\ for the correctness of insertions). The proof is specific to the Citrus tree, and does not seem to use a condition similar to forepassed.
\citet{Brown:2014:GTN} prove a lock-free self-adjusting binary tree, establishing traversal correctness (Lemma 18) based on the fact that nodes that are not in the abstract set are finalized. This is a specific case of our condition (see~\Cref{def:preservation-modification}). Our work distills the ingredients needed for a general proof technique, which can prove traversal correctness for multiple different algorithms.

Program logics for compositional reasoning about concurrent programs and data structures have been studied extensively. In this context, the goal is to define a proof methodology that allows   composing proofs of program's components to get a proof for the entire program, which can also be reused in every valid context of using that program. Improving on the classical Owicki-Gries~\cite{DBLP:journals/cacm/OwickiG76} and Rely-Guarantee~\cite{DBLP:conf/ifip/Jones83} logics, various extensions of Concurrent Separation Logic~\cite{DBLP:conf/popl/BornatCOP05,DBLP:conf/concur/Brookes04,DBLP:conf/concur/OHearn04,DBLP:conf/popl/ParkinsonBO07} have been proposed in order to reason compositionally about different instances of fine-grained concurrency, e.g.~\cite{DBLP:journals/jfp/JungKJBBD18,DBLP:journals/pacmpl/JungLPRTDJ20,DBLP:conf/popl/Ley-WildN13,DBLP:conf/ecoop/PintoDG14,DBLP:conf/esop/SergeyNB15,DBLP:journals/pacmpl/Nanevski0DF19,DBLP:conf/esop/RaadVG15,DBLP:conf/icfp/TuronDB13,conf/cav/DragoiGH13,conf/vmcai/Vafeiadis09,phd/Vafeiadis08,DBLP:journals/pacmpl/KrishnaSW18}. 
However, they focus on the reusability of a proof of a component in a larger context (when composed with other components) while our work focuses on simplifying the proof goals that guarantee linearizability. The concurrent reasoning needed for our framework could be carried out using one of these logics. It is interesting to note that the lazy list has played an important case study in several of these works~\cite{phd/Vafeiadis08,conf/ppopp/VafeiadisHHS06}, and recently some works~\cite{DBLP:conf/pldi/KrishnaPSW,DBLP:journals/pacmpl/KrishnaSW18} have used the Edgeset framework~\cite{DBLP:journals/tods/ShashaG88} for abstracting some of the reasoning.

Some works~\cite[e.g.][]{conf/tacas/AbdullaHHJR13,conf/cav/AmitRRSY07} attempt at more automatic verification of concurrent data structures.
However, they apply in cases where the linearization point of every invocation is \emph{fixed} to a particular statement in the code.
 This is not the case in the algorithms considered in this paper where for instance, the linearization point of \code{contains(k)} invocations is not fixed. Generic reductions of linearizability to assertion checking~\cite[e.g.][]{conf/esop/BouajjaniEEH13,DBLP:conf/icalp/BouajjaniEEH15,DBLP:conf/cav/BouajjaniEEM17,DBLP:conf/pldi/LiangF13,conf/concur/HenzingerSV13,DBLP:conf/cav/Vafeiadis10,DBLP:conf/cav/ZhuPJ15} apply also to algorithms with non-fixed linearization points, but they do not provide a systematic methodology for proving the assertions, which is the main focus of our paper.  %
\section{Conclusion}
\label{sec:conclusion}
In this paper we have presented a simple and effective method for proving traversal correctness and showed its applicability to several complex concurrent search data structures. Our main observation is that proving traversal correctness is possible by analyzing the effect of writes on (static) reachability, guaranteeing a consistent extension of the traversal in each point, while relying on the local nature of the decisions made by traversals in each step.
In a sense, this result demonstrates, surprisingly, that extremely sophisticated concurrency techniques can be tamed using general, comprehensible principles. We hope that this can direct exploration of new algorithms in the design space.
Moving forward to even more intricate data structures, it would be useful to explore general concepts for elements of correctness beyond traversals, such as synchronization patterns in lock-free algorithms.

\begin{changebar}
We have demonstrated the applicability of our proof framework in simplifying ``pen and paper'' proofs. Leveraging our proof argument in a mechanized proof as important future work. %
The technique of~\citet{DBLP:conf/pldi/KrishnaPSW} seems a promising starting point. It does not currently handle unfixed linearization points, which are required in the algorithms we consider. It will also be interesting to see how our proof technique and its decomposition to traversal correctness will interact with new ideas in concurrent separation logics, such as the use prophecy variables in Hoare-style proofs~\cite{DBLP:journals/pacmpl/JungLPRTDJ20}, to successfully mechanize the proofs of the challenging algorithms we consider.
Such mechanization would also need to tackle the need to reason about reachability invariants. Our experience is that the specific forms of reasoning required by applications of our framework usually benefit from the local nature of the modifications. For example, calculating which locations suffered a $k$-reachability reduction is typically obtained from the premise that the location of modification is $k$-reachable, sometimes employing a few relatively simple invariants (e.g. that a list is sorted). Translating these arguments into a mechanized proof will be interesting future work.
\end{changebar}  
\section*{Acknowledgments}
We thank the anonymous referees for their helpful comments.
This research was partially supported by the European Union's Horizon 2020 research and innovation program (grant agreement No.\ 678177, 724464, and 759102), the Spanish MICINN project BOSCO (PGC2018-102210-B-I00), the United States-Israel Binational Science Foundation (BSF) grant No. 2016260, the Len Blavatnik and the Blavatnik Family foundation, the Blavatnik Interdisciplinary Cyber Research Center at Tel Aviv University, the Pazy Foundation, and the Israel Science Foundation (ISF) grant No. 1996/18, 2005/17, and 1810/18.

\clearpage
\bibliography{refs}

\iflong
\clearpage
\appendix

\section{Linearizability of the Logical Ordering Tree: Addendum}
\label{sec:lo-appendix}

\begin{figure}[H]
\begin{framed}
\justify

We consider the LO tree operations return values, and leverage the assertions in Figure~\ref{Fi:LO-code} in showing that they correspond to taking effect on the abstract set.

\vspace{7pt}%
\noindent\textbf{Case 1:\ }
When ${\tt contains}(k)$ returns {\rm false} at
line~\ref{Ln:lo-ContainsRetFalseBigger}, $\Past{\QSReachXK{x}{k}} \land
 \QXK[>]{x}{k}$ holds. The {\tt key} field of every node is
immutable, so
$\Past{\QSReachXK{x}{k} \land \QXK[>]{x}{k}}$ holds too
\begin{changebar}
(so the reachability property and the property of the key have held simultaneously).
\end{changebar}
Therefore, during the
execution of ${\tt contains}(k)$ there is a state $\state$ such that $\QSReachXK{x}{k} \land \QXK[>]{x}{k}$ holds of it.
Since $\QSReachXK{x}{k}$ asserts reachability in the \emph{sorted} linked list, we get that
$\state \models \lnot \exists x' \ldotp
\QSReachXK{x'}{k} \land \QXK{x'}{k} \land \QXR[\neg]{x'}$ holds. We conclude
that $k\notin \repfunc(\state)$. %
\vspace{7pt}%
\noindent\textbf{Case 2:\ }
When ${\tt contains}(k)$ returns {\rm false} at
line~\ref{Ln:lo-ContainsRetFalseRem}, $\Past{\QSReachXK{x}{k} \land \QXR{x}}
\land \QXK{x}{k}$ holds. The {\tt key} field of every node is immutable, so
$\Past{\QSReachXK{x}{k} \land \QXR{x} \land \QXK{x}{k} }$ holds too.
Therefore, during the
execution of ${\tt contains}(k)$ there is a state $\state$ such that $\QSReachXK{x}{k} \land \QXR{x} \land \QXK{x}{k}$ holds of it.
By definition of $\QSReachXK{x}{k}$, there is no other $x'$ such
that both $\QSReachXK{x'}{k}$ and $\QXK{x'}{k}$ hold of the state $\state$. We
get that $\state \models \lnot \exists x' \ldotp
\QSReachXK{x'}{k} \land \QXK{x'}{k} \land \QXR[\neg]{x'}$ holds. We conclude that
$k \notin \repfunc(\state)$. %
\vspace{7pt}%
\noindent\textbf{Case 3:\ }
When ${\tt contains}(k)$ returns {\rm true} at
line~\ref{Ln:lo-ContainsRetTrue}, $\Past{\QSReachXK{x}{k} \land \lnot\QXR{x}}
\land \QXK{x}{k}$ holds. The {\tt key} field of every node is
immutable, so $\Past{\QSReachXK{x}{k}
\land \lnot\QXR{x} \land \QXK{x}{k} }$ holds too. Therefore,
during the execution of ${\tt contains}(k)$ there is a state $\state$ such
that $\state \models \exists
x \ldotp \QSReachXK{x}{k} \land \QXK{x}{k} \land \QXR[\neg]{x}$ holds.
We conclude that $k \in \repfunc(\state)$. %
\end{framed}
    \caption{\footnotesize Proving linearizability from the assertions of the LO tree (part 1/3, $\code{contains}$).}
    \label{lo:assertions-imply-linearizability-contains}
\end{figure}

\begin{figure}[H]
\begin{framed}
\justify
There are only two ways the {\tt insert} operation can take effect. %

\noindent\textbf{Case 4:\ }
When ${\tt insert}(k)$ returns {\rm false} at line~\ref{Ln:lo-insert-ret-false},
$\QSReachXK{s}{k} \land \QXK{s}{k} \land \QXR[\neg]{x}$ holds of the current state $\state$.
It is easy to see that $k \in \repfunc(\state)$, so the abstract set
already contains $k$. %
\vspace{7pt}
\noindent\textbf{Case 5:\ }
When ${\tt insert}(k)$ returns {\rm true} at the end of the operation,
we let $\state$ and $\state'$ be the states before and after
line~\ref{Ln:lo-insert-write-previous}. Since $\state$ satisfies the assertion
at line~\ref{Ln:lo-insert-write}, $\QSReachXK{p}{k} \land
\QXR[\neg]{p} \land\ p.\fnext=s \land k \in (p.\fkey,s.\fkey) \land \QXK{n}{k}
\land \QXR[\neg]{n} \land n.\fnext=s$ holds of $\state$.
It is easy to see that
$\Past{\QSReachXK{s}{k}} \land \QXK[>]{s}{k}$ holds.
Analogously to Case~1, we can conclude
that $k\notin \repfunc(\state)$.
Moreover, since $\QXK{n}{k} \land
\QXR[\neg]{n}$ holds of $\state$, as a result of
line~\ref{Ln:lo-insert-write-previous}, we get that $\state' \models \exists n
\ldotp \QSReachXK{n}{k} \land \QXK{n}{k} \land \QXR[\neg]{n}$ holds and,
therefore, $k \in \repfunc(\state')$. %

It remains to show that for each $k' \neq k$, that $k' \in \repfunc(\state)$ iff
$k' \in \repfunc(\state')$. It is easy to see that ${\tt insert}(k)$ does not affect reachability of nodes preceding $n$ via $\fnext$-links, and line~\ref{Ln:lo-insert-write-previous} ensures that subsequent nodes also remain reachable.
Overall, $\repfunc(\state') = \repfunc(\state) \cup \set{k}$ and $\repfunc(\state) \neq \repfunc(\state')$.

\end{framed}
    \caption{\footnotesize Proving linearizability from the assertions of the LO tree (part 2/3, $\code{insert}$).}
    \label{lo:assertions-imply-linearizability-insert}
\end{figure}

\begin{figure}[H]
\begin{framed}
\justify
There are only two ways the {\tt delete} operation can take effect. %

\noindent\textbf{Case 6:\ }
When ${\tt delete}(k)$ returns {\rm false} at line~\ref{Ln:lo-remove-ret-false},
$\QSReachXK{s}{k} \land \QXK[>]{s}{k}$ holds of the current state
$\state$. Since $\QSReachXK{s}{k}$ asserts reachability in the \emph{sorted}
linked list, we get that $\state \models \lnot \exists x' \ldotp
\QSReachXK{x'}{k} \land \QXK{x'}{k} \land \QXR[\neg]{x'}$ holds.
We conclude that $k \notin \repfunc(\state)$. %
\vspace{7pt}
\noindent\textbf{Case 7:\ }
When ${\tt delete}(k)$ returns {\rm true}, we let $\state$ and $\state'$
be the states before and after line~\ref{Ln:lo-delete-write-rem-mark}. Since
$\state$ satisfies the assertion at line~\ref{Ln:lo-delete-write-rem},
$\QSReachXK{s}{k} \land \QXK{s}{k} \land \QXR[\neg]{s}$ holds of $\state$, so we
immediately get $k \in \repfunc(\state)$. Moreover, once
line~\ref{Ln:lo-delete-write-rem-mark} marks $s$ as removed, we get that 
$\QSReachXK{s}{k} \land \QXK{s}{k} \land \QXR{s}$ holds of $\state'$.
Analogously to Case~2, we can conclude $k \notin \repfunc(\state')$.

It remains to show that for each $k' \neq k$, that $k' \in \repfunc(\state)$ iff
$k' \in \repfunc(\state')$. It is easy to see that ${\tt delete}(k)$ does not affect reachability of nodes preceding $s$ via $\fnext$-links, and~\cref{Ln:lo-delete-write-previous} ensures that subsequent nodes remain reachable.
Overall, $\repfunc(\state') = \repfunc(\state) \setminus \set{k}$ and $\repfunc(\state) \neq \repfunc(\state')$.

\end{framed}
    \caption{\footnotesize Proving linearizability from the assertions of the LO tree (part 3/3, $\code{delete}$).}
    \label{lo:assertions-imply-linearizability-delete}
\end{figure}

\begin{figure}[H]
\begin{framed}
\raggedright

The assertions that do not concern traversals and $\{\Past{P}\}$ properties 
are easy to prove using the following reasoning: %
\begin{itemize}
	\item The \emph{immutability of keys} ensures properties such as $\QXK[>]{x}{k}$, $k \in (p.\fkey,s.\fkey]$, etc.\ in~\cref{Ln:lo-ContainsRetFalseBigger,Ln:lo-ContainsReadRem,Ln:lo-delete-reachable-now,Ln:lo-remove-ret-false,Ln:lo-delete-write-rem,Ln:lo-delete-write} etc. %

	\item Reading fields under \emph{the protection of a lock} ensures properties concerning the $\frem$ field, such as $\QXR[\neg]{p}$ in~\cref{Ln:lo-insert-range}, %
	and properties concerning $\fnext$ such as $p.\fnext=s$ in~\cref{Ln:lo-delete-range}, $p.\fnext=s$, $s.\fnext=y$ in~\cref{Ln:lo-delete-write}, etc.
	
	\item Reasoning about interleavings of the lock-protected critical sections to establish \emph{simple inductive invariants} about the data structure. 
	That when a node
	$p$ is written to by update operations, then $\QSReachXK{p}{k}$ holds (\cref{Ln:lo-insert-reachable-now,Ln:lo-delete-reachable-now}), follows from two invariants of the successors list:
			\begin{itemize} 
				\item If a previously reachable node $p$ is locked and $\QXR[\neg]{p}$ then $\QSReachX{p}$. This invariant is established by considering the possible interleavings of writes, and noting that a node is marked (\cref{Ln:lo-delete-write-rem}) before it is made unreachable in the successors list (\cref{Ln:lo-delete-write-previous}).
			 	\item The successors list is sorted, which follows from the assertions in~\cref{Ln:lo-insert-write,Ln:lo-delete-write}. This implies that if in addition to $\QSReachX{p}$ also $\QXK[\geq]{p}{k}$, then $\QSReachXK{p}{k}$.
			 \end{itemize}
			The same reasoning applies also to $\QSReachXK{s}{k}$ in~\cref{Ln:lo-remove-ret-false,Ln:lo-delete-write-rem,Ln:lo-insert-ret-false}.

			Another invariant justifies $\QXR[\neg]{s}$ in~\cref{Ln:lo-insert-ret-false-unmarked}. The invariant is that every node $s$ that is not locked as part of its deletion such that $\QSReachXK{s}{k}$ is necessarily unmarked, which holds because only $\code{delete}$ removes a node from the successors list, and it marks the node before releasing the lock.

			The proof of these invariants can rely on the assertions by induction (similar to~\Cref{sec:circular}).
\end{itemize} 
Overall, none of these assertions require complex reasoning about the interference of writes with the sequence of reads a traversal performs.

\end{framed}

\caption{\footnotesize Proving non-traversal assertions of the LO tree.}
\label{lo:proving-assertions-illustrations}
\end{figure}

\section{Contention-Friendly Tree with Backtracking}
\label{sec:cf-backtrack}

The code of this algorithm with its assertions, based on~\cite{FeldmanE0RS18}, appears in~\Cref{Fi:CF-BT}. The difference from the algorithm in~\cite{FeldmanE0RS18} is \emph{backtracking}, implemented by pointing from a removed node to its parent in~\crefrange{cft-code:remove-backtrack-start}{cft-code:remove-backtrack-end} (see~\Cref{Fi:Backtracking} for an illustration).

\para{Linearizability and assertions}
Except for traversal correctness, the linearizability proof of this example is exactly as in~\cite{FeldmanE0RS18}. The proof of traversal correctness is different, because the two conditions of the local view argument are violated:
\begin{inparaenum}
	\item the traversal may observe cycles during its run, through the backtracking pointers, which violates temporal acyclicity;
	\item locations are modified even after they are no longer reachable, violating preservation.
\end{inparaenum}

\para{$k$-search paths and $\extend{k}$}
See~\Cref{sec:cf-short} for the definition of $\extend{k}$ and $\QReachXK{\cdot}{k}$.
It follows directly from the definition that $\sigma \models \QReachXK{\ell}{k} \land \extend{k}(\ell,\sigma(\ell),\ell') \implies \sigma \models \QReachXK{\ell'}{k}$.

\begin{figure}[t]
\centering
\begin{tabular}[t]{p{5.6cm}p{5.0cm}p{4cm}}
\begin{lstlisting}
type N 
  int key  
  N left, right 
  bool del,rem

N root$\leftarrow$new N($\infty$);

N$\times$N locate(int k)
  x,y$\leftarrow$root
  while (y$\neq$null $\land$ y$.$key$\neq$k)
    x$\leftarrow$y
    if (x$.$key<k)
      y$\leftarrow$x$.$right
    else
      y$\leftarrow$x$.$left
 $
  \color{assertion}
   {\begin{array}{l}
   \{\PReachXK{x}{k}  \land \PReachXK{y}{k}   \\
    \ \ \land\ \QXK[\neq]{x}{k} \land y \neq \NULL \implies  \QXK{y}{k}\}
  \end{array}
  }
  \label{Ln:RetLocate}
  $  
  return (x,y)

bool contains(int k)  
  (_,y)$\leftarrow$locate(k) 
  if (y = null) 
    $  
    \color{assertion}{\{ \Past{\QReachXK{\NULL}{k}} \} \label{Ln:ContainsRetFalseNULL}}
    $
    return false
  $\color{assertion}{\{ \PReachXK{y}{k} \}  \label{Ln:ContainsBeforeReadDEL}}$
  if (y.del)
    $  
    \color{assertion}{\{ \Past{\QReachXK{y}{k} \land \QXD{y}} \land \QXK{y}{k} \}} \label{Ln:ContainsRetFalseDEL}
    $
    return false
  $  
    \color{assertion}{\{ \Past{\QReachXK{y}{k} \land \neg \QXD{y}}  \land  \QXK{y}{k}  \}}
    \label{Ln:ContainsRetTrue}
  $ 
  return true  
\end{lstlisting}

&

\begin{lstlisting}
bool delete(int k) 
  (_,y)$\leftarrow$locate(k) 
  if (y = null)
    $  
      \color{assertion}{\{ \Past{\QReachXK{\NULL}{k}} \} \label{Ln:DeleteYNull}}
    $  
    return false $\label{Ln:DeleteRetFalseNull}$
  lock(y)
  if (y$.$rem) restart
  ret $\leftarrow$ $\neg$y$.$del
  $
       \color{assertion}{\{\QReachXK{y}{k} \land \QXK[=]{y}{k} \land \QXR[\neg]{y} \}} \label{Ln:DeleteSetDel}%
  $
  y$.$del$\leftarrow$true
  return ret                $\label{Ln:DeleteRetNotDel}$

bool insert(int k)  
  (x,y)$\leftarrow$locate(k)
  $
  %
  \color{assertion}{\{\PReachXK{x}{k}\land \QXK[\neq]{x}{k} \}} \label{Ln:InsertLocate}
  $  
  if (y$\neq$null)
    $
        \color{assertion}{\{\PReachXK{y}{k}\land \QXK[=]{y}{k}\}}
    $
    lock(y)
    if (y$.$rem) restart
    ret $\leftarrow$ y$.$del
    $
        \color{assertion}{\{\QReachXK{y}{k}\land \QXK[=]{y}{k} \land \QXR[\neg]{y}\}} \label{Ln:InsertSetDel} 
    $
    y.del$\leftarrow$false 
    return ret $ \label{Ln:InsertRetDel} $
  lock(x)
  if (x$.$rem) restart
  if (k < x$.$key  $\land$ x.left$=$null)
    $
    \color{assertion}{
    \begin{array}{l}
    \{\QReachXK{x}{k} \land \QXR[\neg]{x}   \\
      \ \ \land\  k < x.key \land x.\mathit{left}=\text{null} \}
    \end{array}}
    \label{Ln:InsertInsertLeft}
    $
    x.left $\leftarrow$ new N(k)
  else if (k > x$.$key  $\land$ x.right$=$null)
    $
    \color{assertion}{
    \begin{array}{l}
    \{\QReachXK{x}{k} \land \QXR[\neg]{x}   \\
      \ \ \land\  k > x.key \land x.\mathit{right}=\text{null} \}
    \end{array}}
    \label{Ln:InsertInsertRight}
    $         
    x.right $\leftarrow$ new N(k)  
  else restart
  return true
\end{lstlisting}

&

\begin{lstlisting}
removeRight()  
  (z,_) $\leftarrow$ locate(*)
  lock(z)
  y $\leftarrow$ z.right
  if(y=null $\lor$ z.rem) 
    return
  lock(y)
  if ($\neg$y.del)
    return
  if (y.left$=$null)
    $\label{Ln:RemoveRight} $ z.right $\leftarrow$ y.right
  else if (y.right$=$null)
    $\label{Ln:RemoveLeft} $ $\label{cft-code:remove-bypass}$z.right $\leftarrow$ y.left
  else
    return
  $\label{cft-code:remove-backtrack-start}$@*\bf{y.right}*@ $\leftarrow$ @*\bf{z}*@
  $\label{cft-code:remove-backtrack-end}$@*\bf{y.left}*@ $\leftarrow$ @*\bf{z}*@
  $\label{cft-code:remove-rem-y}$y.rem $\leftarrow$ true

rotateRightLeft()  
  (p,_) $\leftarrow$ locate(*)
  lock(p)
  y $\leftarrow$ p$.$left
  if(y=null $\lor$ p.rem) 
    return
  lock(y)
  x $\leftarrow$ y$.$left
  if(x=null)
    return
  lock(x)
  z $\leftarrow$ duplicate(y)
  z.left $\leftarrow$ x.right   $\label{Ln:RotateFresh}$ 
  x.right $\leftarrow$ z        $\label{Ln:RotateXRight}$  
  p.left $\leftarrow$ x         $\label{Ln:RotatePLeft}$  $\label{cft-code:rotate-change-parent}$
  y.rem $\leftarrow$ true       $\label{cft-code:rotate-rem-y}$
\end{lstlisting}
\end{tabular} 
\caption{\label{Fi:CF-BT}
Contention-Friendly Tree with backtracking. For brevity, \textbf{unlock} operations are omitted; a procedure releases all the locks it acquired when it terminates or \textbf{restart}s. $*$ denotes an arbitrary key.
}
\end{figure}

\para{Traversals}
It is immediate that the traversal performed by \code{locate} follows extend.

\para{Forepassed}
Deducing assertions of the form $\PReachXK{x}{k}$ (\cref{Ln:RetLocate}).

Consider a traversal in some execution, the effect of concurrent writes on reducing $k$-reachability, and later writes that modify locations whose reachability has been reduced.
We ignore (for now) writes that modify $\frem,\fdel$ fields, since they do not affect reachability.

\begin{itemize}
	\item \textbf{Insert}: Reduces the reachability of $\NULL$ for the key inserted, but $\NULL$ is immutable (\immu). The reachability of other locations is not decreased (\still).

	\item \textbf{Rotate} (see~\Cref{Fi:CF-Rotation}): Writes that modify newly-allocated nodes (such as~\cref{Ln:RotateFresh}) satisfy strong forepassed vacuously, as these were never reachable (\still). \Cref{Ln:RotateXRight} does not decrease $k$-reachability (it increases it for \code{z}) (\still). \Cref{Ln:RotatePLeft} reduces the reachability of \code{y}, but the node has its $\frem$ field set before the procedure releases the lock; later operations would refrain from writing to this object by checking that $\neg\frem$ (\immu).

	\item \textbf{Remove}: \Cref{Ln:RemoveRight} reduces the reachability of \code{y} only (because $\XleftY{\code{y}}{\NULL}$). The only further writes possible to this location are \cref{cft-code:remove-backtrack-start,cft-code:remove-backtrack-end} of the same operation; later operations will refrain from writing due to $\frem$. However, condition (\pastreach) holds: there is a point in time during the traversal where $\code{z}$, to which these writes point, have been reachable. Specifically, let $k$ be s.t.\ $\code{y}$ was $k$-reachable before the write in \cref{Ln:RemoveRight}, but not afterwards. Since the search path went through the link $\code{z}.\fright$, necessarily we had $\QReachXK{\code{z}}{k}$ at the moment before the write in \cref{Ln:RemoveRight}, and the premise is that this write is concurrent to the traversal, so this provides the required moment.

	The case of \cref{Ln:RemoveLeft} is dual to \cref{Ln:RemoveRight}.
\end{itemize}

\para{Reachability + $\fdel$}
Deducing assertions of the form $\Past{\QReachXK{y}{k} \land \QXD{y}},\Past{\QReachXK{y}{k} \land \neg\QXD{y}}$ (\cref{Ln:ContainsRetFalseDEL,Ln:ContainsRetTrue}) through the extension of our framework to reachability with another field (\Cref{sec:reachability-with-field}). To establish these assertions we need to show that after the reachability of $y.\fkey$ is reduced, $y.\fdel$ is not modified. This holds because, as above, when the reachability of a node is reduced, it is marked before releasing the lock, and writes to $\fdel$ first check that the node is unmarked (\cref{Ln:DeleteSetDel}). Note that these are \emph{not} the writes that, in this algorithm, modify a node after its reachability is reduced.
\section{\textsc{Citrus}}\label{app:citrus}

\newcommand{\Since}[2]{#1\mathrel{\mathcal{S}}#2}
\newcommand{\reachleft}[3][\in]{#3#1\mathit{left}^{*}(#2)}

\begin{figure}
\hspace{-3ex}
\begin{tabular}[t]{p{4cm}p{5.3cm}p{5.8cm}}
\begin{lstlisting}
type N 
  int key, data
  N left, right 
  int tag
  bool rem 
  $\mathbf{ghost~int}~\ghost{\code{key}}$

N root$\leftarrow$new N($-1$);
root$.$right$\leftarrow$new N($\infty$);

(bool, N) hasNullChild(N y)
   if (y$.$left$=$null)
      return (y$.$right,true)
   if (y$.$right$=$null)
      return (y$.$left,true)
   return (null,false)

setChild(x,k,t)
  if (k<x$.$key)
    x$.$left$\leftarrow$t 
    if (t$=$null)
      x.tag$\leftarrow$x.tag+1
  else
    x$.$right$\leftarrow$t 
 
N$\times$int$\times$N locate(int k)
  x,y$\leftarrow$root
   
  rcu_read_lock
  while (y$\neq$null $\land$ y$.$key$\neq$k) $\label{LnCitrus:TraverseStart}$
    x$\leftarrow$y
    if (y$.$key<k)
      y$\leftarrow$x$.$right
    else
      y$\leftarrow$x$.$left
  tag$\leftarrow$x$.$tag $\label{LnCitrus:TraverseEnd}$
  rcu_read_unlock
  $
  \begin{array}{l}
    \{\Past{\QReachXKgap{x}{k} \land x.\mathit{tag}=\mathit{tag}} \land \\
    \quad \PReachXKgap{y}{k}   \land \QXK[\neq]{x}{k} \land \\
    \quad \neg \isNull(y) \implies \QXK{y}{k}
    \}
  \end{array} \label{LnCitrus:RetLocate}$
  return (x,tag,y)
\end{lstlisting}

&

\begin{lstlisting}
bool delete(int k) 
  (x,_,y)$\leftarrow$locate(k) 
  if (y$=$null)
    $  
      \{ \Past{\QReachXK{\NULL}{k}} \} \label{LnCitrus:DeleteYNull}
    $  
    return false $\label{LnCitrus:DeleteRetFalseNull}$
  $\{y.\mathit{key}=k\} \label{LnCitrus:DeleteFoundKey}$
  lock(x)
  lock(y)
  if (y$.$rem $\lor$ x$.$rem  $\lor$  $\XparentY[\nin]{x}{y}$)
    restart
  $\{\QReachXK{y}{k} \land \QReachXK{x}{k}\} \label{LnCitrus:DeleteReachableNow}$
  (b,otherChild)$\leftarrow$hasNullChild(y)
  if(b)
    y.rem$\leftarrow$true $\label{LnCitrus:BypassMarkY}$
    $
     \begin{array}{l}
       \{\QReachXK{x}{k} \land \XparentY{x}{y} \land {}\\
       \quad \{y.\fleft,y.\fright\}=\{\NULL,\mathit{otherChild}\}\} 
     \end{array}
     \label{LnCitrus:DeleteOne}
    $
    setChild(x,y$.$key,otherChild) $\label{LnCitrus:DeleteOneBypass}$
    return true
  
  ps$\leftarrow$y; cs$\leftarrow$y.right; nx$\leftarrow$cs.left $\label{LnCitrus:SuccTraversalStart}$
  while(nx$\neq$null)
    ps$\leftarrow$cs;  cs$\leftarrow$nx; nx$\leftarrow$cs$.$left $\label{LnCitrus:SuccTraversalEnd}$
  $\{\PReachXKsucc{ps}{k} \land \PReachXKsucc{cs}{k}\}$ $\label{LnCitrus:SuccTraversalPast}$
  
  lock(ps)
  lock(cs) $\label{LnCitrus:UsefulLock1}$
  if (ps$.$rem $\lor$ cs$.$rem  $\lor$
      (ps$\neq$y $\land$ ps$.$left$\neq$cs) $\lor$ 
      cs.left$\neq$null) 
    restart 
  $\{\QReachXKsucc{ps}{k} \land \QReachXKsucc{cs}{k} \land \XleftY{cs}{\NULL}\}$ $\label{LnCitrus:SuccTraversalNow} \label{LnCitrus:SuccEndNow}$

  $
    \{\forall k'.\, (\QReachXK{y}{k'}  \land k < k' \land \QXKR[<]{cs}{k'}) \implies \QReachXK{\NULL}{k'}\}
    \label{LnCitrus:DeleteWithCopySuccessor} \label{LnCitrus:DeleteSuccessorIntervalEmpty}
  $
  w $\leftarrow$ duplicate(y)   
  w.key $\leftarrow$ cs.key     $\fbox{w.\ghost{\code{key}} \ensuremath{\leftarrow} \code{k}}$ $\label{LnCitrus:DeleteSetWKey}$
  w.data $\leftarrow$ cs.data   $\label{LnCitrus:DeleteSetWData}$
  y.rem$\leftarrow$true        $\label{LnCitrus:DeleteSetYMark}$
  lock(w)
 $
  \begin{array}{l}
  \{\QXK[>]{w}{k} \land (\forall k'.\, \QReachXK{y}{k'}  \land \QXK[<]{w}{k'}) \implies \QReachXYK{w}{y.\iright}{k'}) 
    \land {}\\
   \qquad\qquad\qquad(\forall k'.\, \QReachXK{y}{k'} \land \QXKR[<]{w}{k'}) \implies \QReachXYK{w}{y.\ileft}{k'}) 
    \land {}\\
   \qquad\qquad\qquad(\forall k'.\, \QReachXK{y}{k'} \land k < k' \land \QXKR[<]{w}{k'}) \implies (\QReachXK{y.\iright}{k'} \land\QReachXYK{w}{y.\ileft}{k'}) 
  \} 
  \end{array}
  \label{LnCitrus:DeleteWithCopy}
  $ 
  setChild(x,y$.$key,w) $\label{LnCitrus:DeleteWithCopyBypass}$
  $\{\forall k'.\,(\QReachXK{x}{k'} \land ((\QXKR[<]{x}{k} \land \QXKR[<]{x}{k'}) \lor (\QXK[<]{x}{k} \land \QXK[<]{x}{k'}))) 
  \iff \QReachXK{w}{k'} \}$
  synchronize_rcu $\label{LnCitrus:SyncRcu}$
  cs.rem$\leftarrow$true
  if (ps$=$y)
    $
     \{
      \QReachXK{w}{cs.\ikey} \land w.\iright=cs \land w.\ikey = cs.\ikey \land  \XparentY{ps}{cs}
      %
      %
     \} 
     \label{LnCitrus:DeleteRemoveCSRight}
    $ 
    $
     \{
      \forall k'.\,\QReachXK{cs.\iright}{k'} \implies  (\QReachXK{w}{k'}  \land \QXK[<]{w}{k'}) 
     \} 
     \label{LnCitrus:DeleteRemoveCSRightStructural}
    $ 
    w$.$right$\leftarrow$cs$.$right $\fbox{w.\ghost{\code{key}} \ensuremath{\leftarrow} \code{w.key}}$ $\label{LnCitrus:DeleteRemoveCSRightBypass}$
  else  
    $
     \{
      \QReachXK{w}{cs.\ikey} \land ps.\ileft = cs \land cs.\ileft = \NULL \land w.\ikey = cs.\ikey 
      %
      %
     \} 
     \label{LnCitrus:DeleteRemoveCSLeft}
    $ 
    $
     \{
       \forall k'.\,\QReachXK{cs.\iright}{k'} \implies  (\QReachXK{ps}{k'}  \land \QXKR[<]{ps}{k'})  
      \} 
      \label{LnCitrus:DeleteRemoveCSLeftStructural}
    $ 
    ps$.$left$\leftarrow$cs$.$right  $\fbox{w.\ghost{\code{key}} \ensuremath{\leftarrow} \code{w.key}}$ $\label{LnCitrus:DeleteRemoveCSLeftBypass}$  
    if (cs$.$right $=$ $\NULL$)
      x.tag$\leftarrow$x.tag+1 $\label{LnCitrus:DeleteRemoveCSLeftIncrementTag}$
  return true                $\label{LnCitrus:DeleteRetTrue}$
\end{lstlisting}

&

\begin{lstlisting}
bool insert(int k, int d)  
  (x,tag,y)$\leftarrow$locate(k)
  if (y$\neq$null)
    $
        \{\PReachXK{y}{k}\land \QXK[=]{y}{k}\} \label{LnCitrus:InsertFalse}
    $
    return false $ \label{LnCitrus:InsertRetDel} $
  lock(x)
  if (x$.$rem) restart   
  $\{\Past{\QReachXKgap{x}{k} \land \QXV{x}{\mathit{tag}}}\}$   
  if (k<x$.$key $\land$ x.left$=$null $\label{LnCitrus:InsertIf}$
      $\land$ x.tag$=$tag)            $\label{LnCitrus:InsertIfTag}$
    $\{\QReachXK{x}{k} \land k<x.\fkey \land x.\ileft=\NULL  \} \label{LnCitrus:InsertLeft}$
    x.left $\leftarrow$ new N(k,d)
  else if (k>x$.$key $\land$  x.right$=$null)
    $
      \{\QReachXK{x}{k} \land k>x.\fkey \land x.\iright=\NULL  \} \label{LnCitrus:InsertRight}
    $      
    x.right $\leftarrow$ new N(k,d)
  else 
    restart
  return true

int$\cup$bool contains(int k)  
  (_,_,y)$\leftarrow$locate(k) 
  if (y$=$null) 
    $  
    \{ \Past{\QReachXK{\NULL}{k}} \} \label{LnCitrus:ContainsRetFalseNULL}
    $
    return false
  $  
    \{ \Past{\QReachXK{y}{k}  \land  \QXK{y}{k}} \} \label{LnCitrus:ContainsRetTrue}
  $ 
  return y$.$data  
\end{lstlisting}

\end{tabular} 
\vspace{-8mm}
\caption{\label{Fi:Citrus}
Citrus tree~\cite{Arbel:2014}. Ghost code is boxed and executed atomically with the preceding statement.
When a procedure terminates or \textbf{restart}s  it releases all the locks it acquired. The predicate $\XparentY{x}{y}$ holds when $\XleftY{x}{y} \lor \XrightY{x}{y}$. 
}
\end{figure} %

\subsection{The Citrus Tree}
\Cref{Fi:Citrus} presents the annotated pseudocode of the \textsc{Citrus} tree~\cite{Arbel:2014}.
The algorithm uses two forms of synchronization: standard locks and RCU locks~\cite{rcu,rcu-thesis,userrcu}.
An RCU lock governs the interaction between reader threads and writer threads in a rather nonstandard way: a reader 
traverses the tree under the protection of  the (non-exclusive) 
\code{rcu\_lock}. Writers synchronize  with   readers by invoking \code{synchronize\_rcu}.
When a writer thread $t$ invokes  \code{synchronize\_rcu} it gets blocked until all the threads which held the \code{rcu\_lock} at the time $t$ called \code{synchronize\_rcu} release  the \code{rcu\_lock}. 
Thus, it is   possible that the execution of the writer is resumed when some threads hold the \code{rcu\_lock}.
However, these threads obtained it while the writer was already blocked on the \code{rcu\_lock}.
To void deadlocks, threads never hold the  \code{rcu\_lock} and a standard lock at the same time.

Every operation starts with a call to \code{locate(k)}, which performs a standard
binary tree search to locate the node with the target
key $k$. The traversal is done under the protection of the RCU lock.
\code{locate(k)}   returns the last link it traverses, $(x,y)$ and the tag of $x$.  
Thus, if $k$ is found, $y.key=k$; if $k$ is not found, $y=\NULL$ and $x$ is the node that would be $k$'s
parent if $k$ were in the tree.
The \code{contains} operation returns immediately after the traversal.
However, \code{insert} and \code{delete} may continue in an attempt to insert, respectively, remove a node.
Arguably, the most interesting aspect of the algorithm is that it can physically remove a node $y$ after deleting it logically (by setting the \code{rem} flag) even if $y$ has two children. As explained in Figure~\ref{fig:citrus:delete},
this is done by creating a new node $w$ with the same key  and value as that of the successor $cs$ of $y$ in the tree and the same children as $y$;  marking $y$ as deleted; and then linking $y$'s parent $x$ to $w$.
At this point, the deleting thread invokes \code{synchronize\_rcu} and when it awakens,  it unlinks $cs$ from its parent $ps$.
The RCU lock is used to ensure that a thread $t$ looking for $cs$'s key would find it even if $t$'s traversal already passed through $y$. Note that the deletion of a node $y$ which has a single child (\cref{LnCitrus:DeleteOneBypass}) does not use the \code{rcu\_lock}.

One of the main challenges the algorithm has to overcome is preventing \code{insert} operations from violating the tree's structural invariants:
Mutation requires holding the node's exclusive lock. 
Thus, there is a gap between the time a thread releases the \code{rcu\_lock} and the time it acquires the standard lock. This gap  may lead, e.g., to  insert a key $k$ to the tree even though the tree already contains a node with that key.
This may happen if after the RCU-protected traversal of the thread $t$ executing an \code{insert(}$k$\code{)} operation and  while it waits for the lock of 
the node $ps$ to which it wishes to link $k$ as a left child, another thread $t'$ inserts a node $cs$ with key $k$ as a child of $ps$ and later on moves $k$ up the tree to facilitate the deletion of another node $y$ with two children (like in Figure~\ref{fig:citrus:delete}). 
As $t$ is not  protected by the RCU, thread $t'$ may unlink $cs$ from $ps$ while $t$ is still waiting for $ps$'s lock.
If $t$ gets the lock at this point, it finds that $ps$ has no left child. Thus, $t$ would continue with the insertion of $k$, violating linearizability since the key $k$ is already included in the key-value map.  
The algorithm  prevents errors due to this ABA problem using \emph{tags}: 
Every time a thread nullifies the  \code{left} field of a node, it increments the node's tag.
The \code{locate(k)} operation returns the last link it traversed $(x,y)$ and the tag of $x$ at some point during the traversal. An insertion attempts to link a node as a left child of $x$  only if $x$ has no left child (\cref{LnCitrus:InsertIf}) and the tag of $x$ has not changed since the time \code{locate(k)} read it (\cref{LnCitrus:InsertIfTag}).
This ensures that  the value of $x$'s left field was null \emph{since} that  time and prevents the aforementioned ABA problem.
Furthermore, it ensures that indeed $k$ should be inserted as a child of $x$.
An insertion of a right child is not affected by such a problem, thus, unlike~\cite{Arbel:2014}, 
we only use a single tag field in every node which counts the number of times the  \code{left}  field of the node was nullified.

\subsection{Proving Linearizability and Structural Invariants}\label{Se:CitrusLin}

Verifying that the citrus tree is linearizable is rather simple  assuming that the  assertions annotating the code hold: For example, by using an \emph{abstraction function}
$\repfunc: H \totalto (\mathbb{N}\partialto\mathbb{N})$ that maps a concrete memory state of the tree, $H$, to
the {\em abstract map} represented by this state, and showing that \code{contain}, \code{insert}, and \code{delete}
manipulate this abstraction according to their specification.
More specifically, we define $\repfunc$ to map $H$ to the set of key-value pairs residing in nodes that are reachable from the root on a valid search path for their key. (Recall that a search path to a key $k$ terminates when it reaches a node with key $k$.)
The aforementioned invariants and the assertions almost immediately imply that for every operation invocation $op$, there exists a state $H$ during
$op$'s execution for which the abstract state $\repfunc(H)$ agrees with $op$'s return value, and so $op$ can
be linearized at $H$.  
We need only make the following observations: 
\begin{compactitem}
\item \code{contains()} and a failed \code{delete()}
or \code{insert()} do not modify the memory, and so can be linearized at the point in time in which the assertions
before their \code{return} statements hold.  
 
\item In the state $H$ in which a successful  \code{insert(k)}  performs a write, the assertions in 
lines~\ref{LnCitrus:InsertLeft} and~\ref{LnCitrus:InsertRight} imply that $k \not\in \repfunc(H)$  ($\QReachXK{\NULL}{k}$) and  
that had $k$ been in the tree it would have been the left ($\QReachXK{x}{k} \land x.\ileft=\NULL$) resp. right child of $x$ ($\QReachXK{x}{k} \land x.\iright=\NULL$).

\item A successful \code{delete(k)} operation may delete a node by executing one of the following write actions:
\begin{compactitem}
\item 
The write in~\cref{LnCitrus:DeleteOneBypass} is justified by  the assertion on line~\ref{LnCitrus:DeleteOne} which ensures that $y$ does not have two children and thus bypassing it by directing its parent to point to $y$'s other child 
removes only $k$ itself from the abstract map.
\item 
The write in~\cref{LnCitrus:DeleteWithCopyBypass} is justified by the assertions on lines~\ref{LnCitrus:DeleteWithCopySuccessor} and~\ref{LnCitrus:DeleteWithCopy} using the same argument as before.
In particular, note that the assertion on \cref{LnCitrus:DeleteWithCopySuccessor} ensures that it bypasses $y$ with a node $w$ which (i) has the same children as $y$ and (ii) whose key $k_w$ is of the successor node $cs$ of $k$ (recall that $k$ is the key of $y$) in the tree.
The latter ensures that the operation does not change any search path except for the one looking for keys 
between $k$ and  $k_w$: 
A search path for a key $k'$ such that $k \leq k'<k_w$ would 
continue to $w$'s left subtree (before it stopped at $y$ or traversed $w$'s right subtree) and 
a search path for a key $k_w$ terminates at $w$ (before the write it reached $cs$).
As  the keys   $k \leq k'<k_w$ are not in the tree, and $k$ was to be removed, the change of the abstract map agrees with 
\code{delete}'s return value. %
\item All other writes performed in \code{delete(k)} operate on freshly allocated nodes which are not reachable, %
and thus do not change the abstract map. 
\end{compactitem} 
\end{compactitem}

Proving that the assertions annotating the code in~\Cref{Fi:Citrus} hold goes hand in hand with the verification of the following invariant: 
\begin{compactitem}
\item all the keys in the left subtree of a node $x$ are strictly smaller than its own key, while  the ones in $x$'s right subtree are greater or equal to it (the weak binary search property descrbied by~\cite{Arbel:2014}),    
\end{compactitem}
and the following temporal property:
\begin{compactitem}
\item if the tree contains nodes $u$ and $v$ such that $v.\ileft=\NULL$ and $v$'s key is the successor of $u$'s key, i.e.,  
$u.\ikey < v.\ikey$ and  no other node in the tree contains a key $k$ such that $u.\ikey < k < v.\ikey$,
then as long as  $v.\ileft$ does not change its value then the tree does not contain any node with a key $k$ such that $u.\ikey < k < v.\ikey$. 
\end{compactitem}
These properties can be established easily using induction and circular reasoning (relying on the correctness of the assertions).

\subsection{Relating ``Weak'' and ``Strong'' Reachability}

As we explained in \Cref{sec:citrus-short}, the traversals in this algorithm (invocations of \code{locate}) are correct w.r.t. the ``weak'' reachability predicate $\QReachXKgap{\cdot}{k}$, which implies that any output $x$ satisfies $\Past{\QReachXKgap{x}{k}}$. 
 The proof of the assertions for 
 \code{contains} (\cref{LnCitrus:ContainsRetFalseNULL,LnCitrus:ContainsRetTrue}), 
 \code{delete} (\cref{LnCitrus:DeleteYNull,LnCitrus:DeleteReachableNow})
 as well as \code{insert} (\cref{LnCitrus:InsertFalse,LnCitrus:InsertLeft,LnCitrus:InsertRight}), needs to deduce the existence of a standard $k$-search path in the tree from a traversal that guarantees $\PReachXKgap{x}{k}$. The following lemmas provide the necessary tools.

The first lemma states that whenever there exists a node with the ghost field \ghost{\code{key}} different from \code{key}, the successor relation it encodes still holds in the subtree.
\begin{lemma}
\label{lem:citrus-succ-subtree}
If $x.\code{key} \neq x.\ghost{\code{key}}$ for some node $x$, then all keys $k \in (x.\ghost{\code{key}},x.\code{key})$ are not present in the sub-tree rooted at $x$, and there exists a node $y\neq x$ such that $\QReachXYK{x}{y}{x.\code{key}} \land \QXK{y}{x.\code{key}}$.
\end{lemma}
\begin{proof}
As long as $x.\code{key} \neq x.\ghost{\code{key}}$, $x.\code{key}$ is the smallest key bigger than $\ghost{\code{key}}$ in the sub-tree rooted at $x$, and $x$ is a copy of a node $cs$ storing the key $x.\code{key}$ which is locked and still present in the sub-tree rooted at $x$ (see Figure~\ref{fig:citrus:delete}).
\end{proof}

The following lemma is useful in reasoning about the possible trajectories of ``weak'' $k$-reachability paths: it says that they can deviate from standard $k$-reachability paths in a single node. Let $\mathit{weakLink}$ denote the relation between locations in a state $\sigma$ defined in~\Cref{sec:citrus-short} that we recall below:
\begin{align*}
\sigma\models\mathit{weakLink}(\ell,\ell') \mbox{ iff }&\mbox{$\ell = (o,\code{key})$, $\sigma(o,\code{key})=m$, $\sigma((o,\ghost{\code{key}}))=\ghost{m}$, and }\nonumber \\
&\hspace{1cm}\ell'=(o,\code{right})\mbox{ if $(k>\ghost{m} \land k\neq m) \lor (k=m \land \ghost{m}\neq m), and$}\\
&\hspace{1cm}\ell'=(o,\code{left})\mbox{ if $k<m$}\nonumber
\end{align*}

\begin{lemma}
\label{lem:citrus-ghost-path-breakdown}
If $\sigma \models \QReachXKgap{x}{k}$, then there are locations $\ell,\ell'$ s.t.\ $\sigma \models \QReachXK{\ell}{k}\land \mathit{weakLink}(\ell,\ell') \land \QReachXYK{\ell'}{x}{k}$.
\end{lemma}
\begin{proof}
Recall that the predicate $\QReachXK{x}{k}$ holds in a state $\sigma$ if and only if there is a sequence of
locations $\ell_1,\ldots,\ell_{n+1}$ starting from $\ell_1 = (\rootobj,\fkey)$, ending in $\ell_{n+1} = (x, \fkey)$, and connected via the \code{left} and \code{right} fields in the following manner:
for every $i\in [0,n]$, if $\ell_i = (o,\code{left})$ or $\ell_i = (o,\code{right})$ with $\sigma(\ell_i)=(o',\fkey)$, then $\ell_{i+1} = (o',\fkey)$, and 
\begin{align}
&\mbox{if $\ell_i = (o,\code{key})$, $\sigma(o,\code{key})=m$, then }\nonumber \\
&\hspace{2cm}\ell_{i+1}=(o,\code{right})\mbox{ if $k>m$, and}\label{eq:strong-reach} \\
&\hspace{2cm}\ell_{i+1}=(o,\code{left})\mbox{ if $k<m$}\nonumber
\end{align}
Therefore, from a node with $\code{key}=m$ and $\ghost{\code{key}}=\ghost{m}\neq m$, $\mathit{weakLink}$ deviates from the relation defined in \Cref{eq:strong-reach}  
only in going to the right-subtree in more cases, and for keys in the interval $(\ghost{m},m]$. We want to prove that after one such deviation, the path cannot take another one. Hence,
it suffices to prove that if $x.\code{key}=m$ and $x.\ghost{\code{key}}=\ghost{m}$ and $y.\code{key}=m'$ and $y.\ghost{\code{key}}=\ghost{m'}$, 
and $y$ is in the right-subtree of $x$, then $(\ghost{m},m] \cap (\ghost{m'},m'] = \emptyset$. This follows from the fact that $x.\ghost{\code{key}}\neq y.\ghost{\code{key}}\land x.\code{key}\neq y.\code{key}$ is an invariant of this algorithm, and $x.\code{key}$, resp., $y.\code{key}$, is the smallest key bigger than $x.\ghost{\code{key}}$, resp., $y.\ghost{\code{key}}$.
\end{proof}

These two lemmas imply that $\exists x. \ \PReachXKgap{x}{k} \land \QXK{x}{k} \implies \exists x. \ \PReachXK{x}{k} \land \QXK{x}{k}$, and $\PReachXKgap{\NULL}{k} \implies \PReachXK{\NULL}{k}$. Note that it is not necessarily true that a real $k$-search path existed for \emph{intermediate} locations in a traversal---hence the ``weak'' $k$-reachability predicate---but does hold for the endpoints (finding the key / $\NULL$).
If a search for a key $k$ is affected by the ghost fields (otherwise they hold trivially), then let $\ell,\ell'$ be as in~\Cref{lem:citrus-ghost-path-breakdown}. In particular, $\mathit{weakLink}(\ell,\ell')$ holds and $k\in (\ghost{m},m]$. From \Cref{lem:citrus-succ-subtree}, the path from $\ell'$ must end in either $\NULL$ if $k\neq m$ or the $x$ with $\QXK{x}{m}=k$. In both cases, a standard $k$-search path exists at the point of writing the ghost field \ghost{\code{key}} to the node that contains $\ell$: for the former, the keys in $(\ghost{m},m)$ were not present in the tree at that point (since $m$ is the successor of $\ghost{m}$ in the tree---see~\cref{LnCitrus:DeleteSuccessorIntervalEmpty}); for the latter, $\ell$ was $m$-reachable at that point.

This is enough to prove the assertions at~\cref{LnCitrus:ContainsRetFalseNULL,LnCitrus:ContainsRetTrue} in \code{contains}, and those that justify a failed \code{insert} (\cref{LnCitrus:InsertFalse}) and a failed \code{delete} (\cref{LnCitrus:DeleteYNull}).

\subsection{Inferring Current State Reachability from Reachability in the Past}

The algorithm utilizes \emph{version numbers} (the \code{tag} fields) to ensure that nodes are reachable when modified. We show that the traversal that reads the version number guarantees $\Past{\QReachXKgap{x}{k} \land \QXV{x}{v}}$, and that $\QReachXK{x}{k}$ can be inferred in certain conditions by checking the version number.

\para{Strong Forepassed Condition}
We instantiate the extension for reachability with another field presented in \Cref{sec:reachability-with-field} for the reachability predicate $\QReachXKgap{x}{k}$ and the field $x.\fversion$. Every write $w$ by the algorithm either
\begin{inparaenum}
	\item does not reduce the reachability of $x$, 
	\item no write after $w$ would modify $x.\fversion$.
\end{inparaenum}
The only writes that reduce the reachability of $x$  happen after $x$ is marked as removed, and future writes to $x.\fversion$ first validate that $x$ is not marked. In particular, note that a traversal reads the $\fversion$ field \emph{before} \code{rcu\_read\_unlock}.

\para{Inferring Current State Reachability Using Tags}
The following lemma shows that the assertion in~\cref{LnCitrus:RetLocate} and the additional validations before~\cref{LnCitrus:InsertLeft} imply that the assertion at that line holds. 
\begin{lemma}
The following property holds in any state of the algorithm:
$$
		\Past{\QReachXKgap{x}{k} \land \QXV{x}{v}}
	\land
		x.\fleft = \NULL \land \QXV{x}{v} \land \QXR[\neg]{x} \land \unlocked{x}
	\implies
		\QReachXK{x}{k}.
$$
\end{lemma}
\begin{proof}
This is proved by induction on the number of interfering writes between the moment in the past where $\QReachXKgap{x}{k} \land \QXV{x}{v}$ holds and the current time.

For the base case (no interfering writes), if $\QReachXK{x}{k} \land \QXV{x}{v}$ holds in the past, then $\QReachXK{x}{k}$ follows easily. 
Otherwise, the traversal ending in $x$ must go through a node $y$ with $y.\code{key}=m$ and $y.\ghost{\code{key}}=\ghost{m}\neq m$, and 
$k \in (\ghost{m},m]$. By \Cref{lem:citrus-succ-subtree,lem:citrus-ghost-path-breakdown}, $x$ can only be a node holding the key $m$ which is however locked (locked at \cref{LnCitrus:UsefulLock1} by a concurrent removal of a node with 2 children). This implies that the claim holds vacuously.

For the induction step, in the last preceding time when $x$ was unlocked, $x$ had also a $\NULL$ left child, since writes that set the left child to $\NULL$ increase the version number (calls to \code{setChild} in~\cref{LnCitrus:DeleteOneBypass,LnCitrus:DeleteRemoveCSLeftIncrementTag}). It was also unmarked, and hence $\QReachXK{x}{k}$. What writes could have modified this fact? 
\begin{itemize}
	\item A bypass (\cref{LnCitrus:DeleteOneBypass}), but then $x$ would have been marked before the locks were released (\cref{LnCitrus:BypassMarkY}).
	\item the removal of a node with 2 children (\cref{LnCitrus:DeleteWithCopyBypass}), but then $x$ would either be still locked or marked (recall that $x.\fleft=\NULL$).
\end{itemize}
\end{proof}

\para{Insertion as a right child}
Insertion to a right child (\cref{LnCitrus:InsertRight}) is performed without checking the version number. The reachability in the current state follows from the following property concerning the reachability of $\fright$ fields only:
$$
\PReachXKgap{x.\fright}{k} \land \QXR[\neg]{x} \land \unlocked{x} \implies \QReachXK{x.\fright}{k}
$$
The reason this holds is that the first step in removing a node with 2 children (\cref{LnCitrus:DeleteWithCopyBypass}) does not reduce reachability of $\fright$ fields, but only $\fleft$ fields (the predicate $\QReachXKgap{x}{k}$ makes it possible to reach more locations by going to the right sub-tree instead of the left sub-tree). %

\para{Delete}
For \code{delete}, the assertion $\QReachXK{y}{k}$ in~\cref{LnCitrus:DeleteReachableNow} is deduced without alluding to version numbers, using a simpler invariant, that unremoved, unlocked nodes are reachable for their own key: $\Past{\QReachXKgap{x}{k}}\land \QXK{x}{k} \land \QXR[\neg]{x} \land \unlocked{x} \implies \QReachXK{x}{k}$. Note that a removal of a node with 2 children generates two copies of the same key, but marks the obsolete one before unlocking it (\cref{LnCitrus:DeleteSetYMark}). $\QReachXK{x}{k}$ follows since $\QXR[\neg]{x}$ and it is an invariant that a node has at most one unmarked, unlocked parent.

\subsection{Searching for the Successor Key}

The traversal in~\code{delete} that finds the successor of the deletion target, in~\crefrange{LnCitrus:SuccTraversalStart}{LnCitrus:SuccTraversalEnd}, differs from the previous traversals, in that it is not performed in the scope of an RCU lock. It also does not utilize version numbers. 
This is justified by observing that the problematic interfering writes, the final stage in removing a node with 2 children---which would have violated the forepassed condition of the standard traversals (see \Cref{sec:citrus-short}) had it been performed without the RCU---is impossible in the successor traversal, because the node \code{y} (with key $k$) is \emph{locked}.

Formally, we consider instances of reachability predicates $\QReachXKsucc{\ell}{k}$, with $\epsilon > 0$ chosen s.t.\ $k+\epsilon < k''$ for every $k''$ that may be present in the tree during the traversal.\footnote{Choosing such $\epsilon$ is trivial when keys are integers, but since this is only an artifact of the proof it can be chosen as prophecy concerning the keys that are present and those that would be inserted.}

The guarantee from the traversal is $\PReachXKsucc{\code{ps}}{k}$ in~\cref{LnCitrus:SuccTraversalPast}. It holds because the traversal starts from $\code{y}$ that has $\QReachXKsucc{\code{y}}{k}$---this follows from $\QReachXK{\code{y}}{k}$ and the fact that $\QXK{\code{y}}{k}$ (\cref{LnCitrus:DeleteFoundKey,LnCitrus:DeleteReachableNow})---and that the \emph{strong forepassed} condition holds: that only writes that reduce $k+\epsilon$-reachability are
\begin{inparaenum}
	\item bypass (including the final stage in removing a node with 2 children), which marks the node, as usual, and
	\item the first step of a \emph{different}\footnote{The same operation would also perform this write later, but this is only \emph{after} the successor traversal completes.} remove2c (\cref{LnCitrus:DeleteWithCopyBypass}), deleting $k_0$ with copying its successor with key $k_1$. However, this only reduces the reachability of nodes in the subtree of said $k_0$ for keys $\tilde{k} \in (k_0,k_1]$. Note that $k \neq k_0$ the node with key $k$ is locked. Hence, similarly to the argument in \Cref{lem:citrus-ghost-path-breakdown}, a different removal of a node with 2 children must operate on a disjoint interval: $k \not\in (k_0,k_1]$ and hence also $k+\epsilon \not\in (k_0,k_1]$. 
\end{inparaenum}

From a search path in the past we obtain $\QReachXKsucc{\code{ps}}{k}$ in the current state in~\cref{LnCitrus:SuccTraversalNow} by a simple invariant: while $\code{y}$ is (continuously) locked, $\PReachXKsucc{\ell}{k} \land \QXR[\neg]{\ell} \implies \QReachXKsucc{\ell}{k}$. The argument is the same as in proving the strong forepassed condition.

The assertions concerning $\QReachXKsucc{\code{cs}}{k}$ (\cref{LnCitrus:SuccEndNow}) are deduced in a straightforward manner from the preceding assertions and reading the link between $\code{ps}$ and $\code{cs}$ under the protection of a lock.  
\fi

\end{document}